\documentclass[11 pt] {article}
\usepackage{graphicx} \usepackage{color}
\usepackage{amssymb,amsfonts}
\usepackage{amsmath,amssymb,amsthm}
\newcommand{\Real}{\mathbb{R}}

\def \b1 {\mathbf{1}}

\theoremstyle{plain} \newtheorem{thm}{Theorem}%[section]
 \newtheorem{lem}[thm]{Lemma}

\theoremstyle{definition} %[section]

\theoremstyle{remark} \newtheorem{rem}{Remark}[section]

\usepackage{fullpage}
\title{Compressed Blind De-convolution of Filtered Sparse Processes\thanks{This research
was supported by NSF CAREER award ECS 0449194}}

\author{V. Saligrama$^{\dagger}$ \quad M. Zhao\footnote{Author Names Appear in Alphabetical Order}\\Department of
Electrical and Computer Engineering\\ Boston University, MA 02215
\\\tt \{srv,\,mqzhao\}@bu.edu}

\date{}

\begin{document}
\maketitle

\begin{abstract}
Suppose the signal $x\in \Real^n$ is realized by driving a k-sparse signal $u \in \Real^n$ through an arbitrary unknown stable discrete-linear time invariant system $H$, namely, $x(t) = (h * u)(t)$, where $h(\cdot)$ is the impulse response of the operator $H$. Is $x(\cdot)$ compressible in the conventional sense of compressed sensing? Namely, can $x(t)$ be reconstructed from small set of measurements obtained through suitable random projections? For the case when the unknown system $H$ is auto-regressive (i.e. all pole) of a known order it turns out that $x$ can indeed be reconstructed from $O(k\log(n))$ measurements. We develop a novel LP optimization algorithm and show that both the unknown filter $H$ and the sparse input $z$ can be reliably estimated.
\end{abstract}

\section{Introduction}

In this paper we focus on blind de-convolution problems for filtered sparse processes. These types of processes naturally arise in reflection seismology~\cite{robinson}. The LTI system $H$ is commonly referred to as the wavelet, which can be unknown, and serves as the input signal. This input signal passes through the different layers of earth and the reflected signal $z$ corresponds to the reflection coefficients from the different layers.  The signal $z$ is typically sparse. The reflected output, which is referred to as the seismic trace, is recorded by a geophone. Other applications of filtered sparse processes include nuclear radiation~\cite{doucet}, neuronal spike trains~\cite{Gerstner} and communications~\cite{snyder}.

Specifically, a sparse input $u(t)$ is filtered by an unknown infinite impulse response (IIR) discrete time stable linear filter $H$ and the resulting output $$x(t) = (Hu)(t) = \sum_{i} u(\tau_i) h(t -\tau_i)$$ is measured in Gaussian noise, namely, $y(t) = x(t) + n(t)$ for $t = 0,\,1,\,\ldots,N$. The goal is to detect $z(t)$, and estimate the filter $H$. The main approach heretofore proposed for blind de-convolution involves heuristic iterative block decomposition schemes (first proposed in \cite{mendel}). Here the filter and sparse inputs are alternatively estimated by holding one of them constant. While these algorithms can work in some cases, no systematic performance guarantees currently exist. We explore a convex optimization framework for blind de-convolution.

In addition we consider the compressed sensing problem, namely, $x(t)$ is compressed by means of a random Gaussian filter ensemble, as described in Figure~\ref{fig:block1} and the resulting output is measured noisily. Analogously, we can consider a random excitation model as in Figure~\ref{fig:block2}. Our task is to detect $z(t)$ and estimate $H$. Our goal is to characterize the minimum number of random samples required for accurate detection and estimation.
\begin{figure}[ht]
\begin{centering}
\includegraphics[width = 8.0cm]{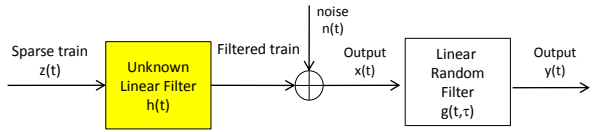}
\caption{\small Compressed blind de-convolution.}
\label{fig:block1}
\end{centering}
\end{figure}
\begin{figure}[ht]
\begin{centering}
\includegraphics[width = 8.0cm]{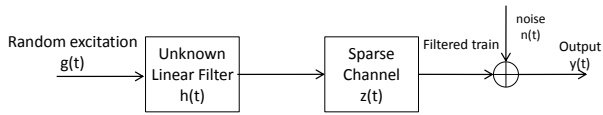}
\caption{\small Estimation of filtered sparse process: Random Excitation.}
\label{fig:block2}
\end{centering}
\end{figure}

%\noindent
%\underline{{\bf
\subsection{Comparison with Compressed Sensing}
Note that this is significantly different from the standard Compressed sensing(CS)~\cite{Candes1,Donoho1} problem. In standard CS we have a signal or image, $x \in \Real^n$, which is sparse in some transform domain. Specifically, there is a known orthonormal matrix $H$ such that the transformed signal $z = H^T x$ is k-sparse, namely, has fewer than $k$ non-zero components\footnote{This is often referred to as transform domain sparsity.}. A matrix $G \in \Real^{m\times n}$ then maps $x$ to measurements $y = Gx = GHu$.  For suitable choices of matrices $G$, such as those satisfying the so called Restricted Isometry Property (RIP), the k-sparse signal $z$ can be recovered with $O(k\log(n))$ measurements as a solution to a convex optimization problem:
$$
\min \|u\|_1 \,\,\, \mbox{subject to} \,\,\,y = GHu
$$
This result holds for all sparsity levels $k \leq \alpha n,\,\alpha <1$, for sufficiently small $\alpha$. There has been significant effort in CS in recent years leading to various generalizations of this fundamental result. This includes the case when the signal $x$ is approximately sparse (see \cite{Donoho2,Candes3}) and when the measurements are noisy, i.e., $y = GHu+e$ (see ~\cite{Candes3}).

This paper is a significant extension of CS to cases where $H$ is not only not orthonormal but also arbitrary and unknown. Specifically, $H$, is a causal discrete linear time invariant system (LTI) with an unknown impulse response function $h(\cdot)$ as described above. A typical signal $x$ is neither sparse nor approximately sparse as we will see in Section~\ref{sec:simulation}.
%The signal $x$ can also be written in matrix notation as $x= Hz$
%where, $H$ is a Toeplitz matrix (possibly infinite).
%Here the signal $x$ has been generated by passing a 10-sparse signal through a third order system.
%\begin{figure}[ht]
%\begin{centering}
%\includegraphics[width = 8.0cm]{../ARprocess/AR3.eps}
%\caption{\small Output of a sparse input passing through an autoregressive process of order $3$ with
%poles $\alpha_1=0.9,\,\alpha_2=0.7$ and $\alpha_3=0.2$.}
%\label{fig:AR3}
%\end{centering}
%\end{figure}

\subsection{Our Approach}
Our CS problem (schematically shown in Figures~\ref{fig:block1}~\ref{fig:block2}) boils down to determining whether there is a sampling operator $G$ with $O(k \log(n))$ samples such that the signal $x$ can be recovered uniquely from the samples $y=Gx = GHu$ using a convex optimization algorithm. It turns out that this is indeed the case when $H$ is belongs to the set of stable finite dimensional AR processes of a known order.

At first glance the problem as posed appears difficult. For one there is no reason $GH$ satisfies isometry property when $H$ is not orthonormal. To build intuition we describe a practically relevant problem. A specific example is when $x$ is a
one-dimensional piecewise constant signal. Such a signal is not sparse but does have a sparse derivative, namely,  $u(t)=x(t)-x(t-1)$ is sparse. Clearly, the signal $x$ can represented as an output of an (integral) operator $H$ acting on a sparse input $u$, namely, $x=Hu$. However, $H$ is no longer orthonormal. To account for this scenario one usually minimizes the total variation (TV) of the signal. A compressed sensing perspective for this case has already been  developed~\cite{Candes2}.

We develop an alternative approach here. Suppose we now filter $x$ through an LTI system $G$ whose impulse response is $g(t)$. Mathematically, we have,
$$
y(t) = (g * x)(t) = (g * h *u)(t) = ((g *h)*u)(t) = (h *g*u)(t)
$$
Since, the composite system $g*h$ is LTI we have that,
$$
z(t) \stackrel{\Delta}{=} y(t) - y(t-1) = g*(x(t) - x(t-1)) = (g*u)(t)
$$
Now we are in the familiar situation of $z = Gu$ of the standard CS problem, except that $G$ is a Toeplitz matrix. Consequently, if the Toeplitz matrix $G$ satisfies the RIP property we can recover $z$ using standard tools in CS. Indeed RIP properties of Toeplitz matrices have been studied~\cite{saligrama}. Note that this idea generalizes to arbitrary but known finite dimensional stable LTI systems, $H$. The main idea being used here is the commutative property of convolutions.

However, the question arises as to how to deal with unknown system $H$? It turns out that corresponding to every finite dimensional LTI system there is an annihilating filter~\cite{vetterli}. If $H$ is a pth order linear dynamical system it turns out that the annihilating filter, $H^{\perp}$ is parameterized by $p$ parameters. Now employing commutativity of convolution, namely, $g *h = h*g$, followed by filtering through the annihilator we are left with a linear characterization of the measurement equations.
We are now in a position to pose a related $\ell_1$ optimization problem where the parameters are the sparse signal $z$ as well as the parameters governing the annihilating filter. Our proof techniques are based on
duality theory.

Strictly speaking, for AR models commutativity is not necessary. Indeed, we could consider general random projections, but this comes at a cost of increasing the number of measurements as we will see later. On the other hand RIP properties for random projections is (provably) significantly stronger than Toeplitz matrices. Nevertheless note that in the random excitation scenario of Figure~\ref{fig:block2}, the structure does not lend itself to a random projection interpretation. For these reasons we consider both constructions in the paper.

The paper is organized as follows. The mathematical formulation of
the problem is presented in Section \ref{sec:setup}. Section
\ref{sec:alg} describes the new $\ell_1$ minimization  algorithm. The result for recovery with AR filtered processes (Theorem \ref{thm1}) is stated in this
section. The proof of Theorem \ref{thm1} can be found in Section
\ref{sec:proof}. To help the reader understand the main idea of the
proof we first consider a very simple case and Section
\ref{sec:general} provides the proof for the general case. Section
\ref{sec:blind} addresses the blind-deconvolution problem, which can
be regarded as a noisy version of our problem. We use LASSO to solve
this problem and the detailed proof is provided in Section
\ref{sec:proof2}. In Section \ref{sec:ext}, we extend the our
techniques to two related problems, namely, decoding of ARMA process
and decoding of a non-causal AR process. Finally, simulation results
are shown in Section \ref{sec:simulation}.
%The proof of Theorem \ref{thm1} can be found in Section
%\ref{sec:proof}. To help the reader understand the main idea of the
%proof we first consider a very simple case and Section
%\ref{sec:general} provides the proof for the general case.

\section{Problem Set-up}\label{sec:setup}
Our objective is to reconstruct an autoregressive (AR)
process $x(t)$ from a number of linear and non-adaptive
measurements. An autoregressive model is known as an ``all-pole"
model, and has the general form
\begin{eqnarray}\label{eq:AR}x(n) +\sum_{i=1}^p a_i
x(t-i)=z(t)\end{eqnarray} where $z(t)$ is a sparse driving process. We assume the vector $z=[z_0,\cdots,z_{n-1}]^T$ is $k$-sparse, that
is, there are only $k$ non-zero components in $z$. The task of
compressed sensing is to find the AR model coefficients
$a=[a_1,\cdots,a_p]^T$ and the driving process
$z=[z_0,\cdots,z_{n-1}]^T$ from the measurement $y$. In this paper, we assume that the AR process $x(t)$ is
stable, that is, the magnitude of all the poles of the system is
strictly smaller than 1. In later discussion, we use $x_t$ or $x(t)$ interchangeably for
convenience of exposition.

Note that in standard CS setup, the signal $x$ is assumed to be
sparse in some \emph{known} transform space. However, in our
problem, the AR model is assumed to be unknown and the main
contribution of this paper is to solve this new problem efficiently.

We consider two types of compressed sensing scenarios:

\subsection{\bf Toeplitz Matrices} Here we
realize $m$ measurements by applying the sensing matrix $G$ to
signal $x = [x_0,\cdots,x_{n-1}]^T$.
\begin{eqnarray}\label{eq:sense}
\begin{bmatrix}
y_0\\y_1 \\
\vdots\\
y_{m-1}
\end{bmatrix}
=\begin{bmatrix}
g_{n-m} & g_{n-m-1} & \cdots & g_0 & \cdots & 0 & 0\\
g_{n-m+1} & g_{n-m} & \cdots &g_1 & g_0 & \cdots & 0\\
\vdots & \vdots & \ddots & \vdots & \ddots & \ddots & 0\\
g_{n-1} & g _{n-2} & \cdots & \cdots & \cdots & g_1 & g_0
\end{bmatrix}
\begin{bmatrix}
x_0\\x_1\\\vdots\\x_{n-1}
\end{bmatrix}
\end{eqnarray}
where each entry $g_{i}$ is independent Gaussian random variable
$\mathcal{N}(0,1)$ or independent Bernoulli $\pm 1$ entries. Here
the Toeplitz matrix $G$ preserves the shift-structure of the signal. Roughly speaking,
assume $z'$ is a shifted version of $z$ (disregarding the boundary
effect), then $Gz'$ is also just a shifted version of $Gz$. This is particularly suitable for the random excitation model of Figure~\ref{fig:block2}.

For notational purposes we denote by $x^{[s]}$ (or $G^{[s]}$) to
denote the subvector of $x$ (or submatrix of $G$) that is composed
of the last $s$ components (or $s$ rows) of $x$ (or $G$). By rearranging the above Equation \ref{eq:sense} and using the shift-property of $G$, we have the
following equation.
\begin{eqnarray}\label{eq:mult}
Y =
\begin{bmatrix}
y_{p} &y_{p-1}&\cdots& y_1 & y_0 \\
y_{p+1} &y_{p}&\cdots& y_{2} & y_1 \\
\vdots &\vdots&\ddots&\vdots&\vdots\\
y_{m-1} &y_{m-2}&\cdots& y_{m-p} & y_{m-p-1} \\
\end{bmatrix}
\begin{bmatrix}
1\\
a_1\\
\vdots\\
a_{p}
\end{bmatrix}
= G^{[m-p]}z
\end{eqnarray}
where we recall that $z= [z_0,\cdots,z_{n-1}]^T$. Now Equation \ref{eq:mult} is simplified to
\begin{eqnarray}\label{eq:model}Ya+y^{[m-p]}=G^{[m-p]}z\end{eqnarray}
where $a=[a_1,\cdots,a_p]^T\in \mathbb{R}^p$ and $z\in\mathbb{R}^n$
($k$-sparse) need to be decoded from the model.

\subsection{\bf Random Projections} Here we consider randomly projecting the raw measurements $x(t)$, namely,
$$
y(t)= \sum_{\tau=0}^{n-1} g_{t,\tau} x(\tau),\,\, t=0,\,1,\,\ldots,\,m
$$
where, each entry $g_{t,\tau}$ is an independent Gaussian random variable
$\mathcal{N}(0,1)$ or independent Bernoulli $\pm 1$ entry. The reason for choosing random projections over random filters is that IID random Gaussian/Bernoulli matrix ensembles have superior RIP constants. The optimal RIP constants for toeplitz constructions has not been fully answered. Nevertheless, note that to form the matrix $Y$ with random projections requires significantly more projections. This is because we can no longer exploit the shift-invariant property of convolutions. For instance, consider again the matrix $Y$ of Equation~\ref{eq:mult} above: if random projections were employed instead of Toeplitz construction the entry $y_1$ on row 1 will not be equal to the entry $y_1$ in the second row. This means that for a $p$th order model we will require $m\times p$ measurements. \\

{\noindent \bf Notation: } To avoid any confusion, we use $u^*$ to
denote the true spike train and $u$ refers to any
possible solution in the decoding algorithm. Similarly, $a^*$ represents the true
coefficients.

\section{$\ell_1$-minimization Algorithm for AR Models}\label{sec:alg}

Since the AR model is unknown, standard decoding algorithms (e.g.,
Basis Pursuit \cite{Donoho2}, OMP \cite{Tropp3}, Lasso \cite{Wain1},
etc.) can not be directly applied to this problem. However, we can
regard the signal $(u,a)$ (the original signal $u$ together with the
unknown coefficients $a$) as the new input to the model and $(u,a)$
is still sparse if $p$ (the length of $a$) is small.

With this in mind we solve the following $\ell_1$
minimization algorithm
\begin{eqnarray}\label{eq:form1}
\min_{u\in\mathbb{R}^n,a\in\mathbb{R}^p}\quad \|u\|_1 \,\,\text{
subject to}\quad Ya+y^{[m-p]}=G^{[m-p]}u
\end{eqnarray}

More generally, when the measurement $y$ is contaminated by noise,
that is, the sensing model becomes $y=Gx+w$ where $w$ is Gaussian
noise, the above LP algorithm will be replaced by Lasso,
\begin{eqnarray}\label{eq:Lasso}
\min_{u\in\mathbb{R}^n,a\in\mathbb{R}^{p}}\quad \frac{1}{2}\|Ya+
y^{[m-p]}-G^{[m-p]}u\|_2^2+\lambda\|u\|_1
\end{eqnarray}
where $\lambda$ is a tuning parameter that adapts to the noise
level.

Alternatively, the coefficient $a$ can be solved from Equation
\ref{eq:model} by taking pseudo-inverse of $Y$,
\begin{eqnarray}\label{eq:pseudo}
a = (Y^T Y)^{-1}Y^T\left(G^{[m-p]}u-y^{[m-p]}\right)
\end{eqnarray}
Then Equation \ref{eq:model} becomes
\[(I-Y(Y^T Y)^{-1}Y^T)y^{[m-p]}=(I-Y(Y^T
Y)^{-1}Y^T)G^{[m-p]}u\] and similar to Equation \ref{eq:form1} we
can apply the following $\ell_1$ minimization to find the solution
for $u$.
\begin{eqnarray}\label{eq:BP}
\min_{u\in\mathbb{R}^n}\quad \|u\|_1 \,\,\text{ subject to}\quad
Py^{[m-p]}=PG^{[m-p]}u
\end{eqnarray}
where $P$ denotes the projection matrix $I-Y(Y^T Y)^{-1}Y^T$ and
$\|u\|_1$ denotes the $\ell_1$ norm of $u$. Suppose the solution of
Equation \ref{eq:BP}  is $\hat u$. Then $a$ can be easily derived
by $\hat a = (Y^T Y)^{-1}Y^T\left(G^{[m-p]}\hat u-y^{[m-p]}\right)$
and the signal $x(n)$ can be recovered through Equation \ref{eq:AR}.

We note that Equation \ref{eq:BP} is equivalent to Equation
\ref{eq:form1} if $Y^T Y$ is invertible, which is always assumed to
be true in this paper. To summarize the above discussion, our
algorithm is summarized below.

\noindent
{\bf (1) Inputs:} Measurement $y$, sensing matrix $G$ and order of the system $p$.\\
{\bf (2) Compute $u$ and $a$:} Solve the $\ell_1$ minimization
(Equation \ref{eq:form1} or \ref{eq:BP}) or Lasso (Equation \ref{eq:Lasso}).\\
%{\bf (3) Compute $a$:} Find the model coefficients $a$ via taking pseudo-inverse (equation \ref{eq:pseudo}).\\
{\bf (3) Reconstruction:} Recover the signal $x(n)$ through forward propagation of the AR
model of Equation \ref{eq:AR}.\\

Before stating the main result, we recall that for
every integer $S$ the {\em restricted isometry constant}
\cite{Candes3,Candes5}, $\delta_S$ is defined to be the smallest
quantity such that $G^{[m-p]}_T$ obeys
\begin{eqnarray}\label{eq:RIP}
(1-\delta_S)\|x\|_2^2\leq \|G^{[m-p]}_Tx\|^2_2\leq
(1+\delta_S)\|x\|^2_2
\end{eqnarray}
for all subsets $T\subset\{0,1,\cdots,n-1\}$ of cardinality at most
$S$ and all $(x_j)_{j\in T}$.

Note that when the AR filter
$a(n)$ is known the result is a direct application of standard compressed
sensing results. We state this without proof below for the sake of completion.
In other words, if the coefficients $a(n)$ are known, $u^*(\cdot)$ is the true driving process in Equation~\ref{eq:AR} then $u^*(\cdot)$ is the unique minimizer of
\begin{eqnarray}
\min_{u\in\mathbb{R}^n,a\in\mathbb{R}^p}\quad \|u\|_1 \,\,\text{
subject to}\quad Ya+y^{[m-p]}=G^{[m-p]}u
\end{eqnarray}

%\begin{thm}\label{thm0}
%Suppose integer $S$ satisfies $\delta_{S}+\delta_{2S}+\delta_{3S}<
%1$. Suppose the driving process $u^*(n)$ is $k$-sparse and satisfies
%$k\leq S$. If the coefficients $a(n)$ is known, then $u^*$ is the
%unique minimizer of
%\begin{eqnarray}
%\min_{u\in\mathbb{R}^n,a\in\mathbb{R}^p}\quad \|u\|_1 \,\,\text{
%subject to}\quad Ya+y^{[m-p]}=G^{[m-p]}u
%\end{eqnarray}
%\end{thm}

A main result of our paper is the following where $a(n)$ is assumed
to be unknown. We need the following assumptions before we state our the theorem.

\begin{enumerate}
\item {\bf Constant Order:}
We assume that $p$, the order of AR process $x(n)$, is a constant
(i.e., $p$ does not scale with $n,m,S$).

\item {\bf Exponential Decay:}
Suppose the impulse response $|h(i)|$ of the AR model satisfies
$$|h(i)|\leq M\rho^i$$ for some constant $M$ and $0<\rho<1$.

\item {\bf Distance between Spikes:} We define the constant $l:= \left(\log(\frac{2}{1-\rho})+p\log
(\frac{6\beta_{\max}M}{\beta_{\min}})\right)/\log(\rho^{-1})+p$ and
impose the condition that any two spikes, $u_i^*,u_{j}^*$ satisfy $|i-j|>l,\,\,i \not = j$. This implies that the sparsity $k:=|\text{Supp}(u^*)|\leq \min\{S/l,S/3\}$.

\item {\bf Spike Amplitude:} We also
assume that any spike is bounded, $\beta_{\min}\leq |u_k|\leq
\beta_{\max},\forall k\in\text{Supp}(u^*)$.
\end{enumerate}

\begin{thm}\label{thm1}
Suppose assumptions $1$--$4$ above are satisfied. Let the integer $S$ satisfy $\frac{\delta_S}{1-3\delta_S}< 1$.
If $u^*(\cdot)$ is the true driving process in Equation~\ref{eq:AR} then it is the unique minimizer of
\begin{eqnarray}
\min_{u\in\mathbb{R}^n,a\in\mathbb{R}^p}\quad \|u\|_1 \,\,\text{
subject to}\quad Ya+y^{[m-p]}=G^{[m-p]}u
\end{eqnarray}
\end{thm}

Intuitively speaking, the condition in the theorem requires that the
driving process $u(n)$ is sparse enough and any two spikes
$(u_i,u_{j})$ are reasonably far away from each other. This type of assumption is actually also necessary. In section
\ref{sec:general}, we give an example where two spikes are
consecutive  and show that in this case $x(n)$ can not be solved via
equation \ref{eq:BP}. The proof of Theorem \ref{thm1} is presented
in Section \ref{sec:proof}.

\begin{rem} The reader might be curious as to whether a random convolution train provides benefits over random projection. Note that by using random convolutions we can naturally exploit shift-invariance property. Since $Y \in \Real^{m-p \times p}$ as in Equation~\ref{eq:mult} is a partial Toeplitz matrix, we only need $m$ output measurements. In contrast for a random projection, since we can no longer exploit this property, we would require $O(mp)$ measurements.
\end{rem}

\subsection{Noisy Blind-deconvolution}\label{sec:blind}

We consider the noisy blind-deconvolution problem with IID Gaussian noise, $w_i\sim \mathcal{N}(0,\sigma^2)$, and measurements
\begin{eqnarray}\label{eq:mes}y(n) = x(n) +w(n)\end{eqnarray} where
the process $x(n)$ is modeled by $x(n) +\sum_{i=1}^p a_i
x(n-i)=u(n)$. In this section we consider the problem
of reconstructing the sparse spike train $u(n)$ and coefficients $a$
from the observed signals $y(n)$. This problem is called ``Blind
deconvolution'' \cite{doucet,Baziw} and it is a simplified version
of the Compressed Sensing problem where the sensing matrix $G$ is
identity matrix. To the best of our knowledge, even this simplified
problem is still not completely solved in literature. Therefore, we focus on the uncompressed noisy version here. The noisy compressed version is technically more involved and will be reported elsewhere.

Replacing $x(n)$ with $y(n)-w(n)$ in the AR model, we get
\begin{eqnarray}
y(n) +\sum_{i=1}^p a_i y(n-i)=u(n)+e(n)
\end{eqnarray}
where we denote $e(n) := w(n)+\sum_{i=1}^p a_i w(n-i)$.

Again by introducing $$Y = \begin{bmatrix}
 0 & \cdots & 0\\
 y_0 & \cdots &0\\
  \vdots & \ddots & \vdots\\
y_{p-1} &\cdots & y_0\\
 \vdots & \vdots & \vdots\\
 y _{n-2} & \cdots & y_{n-p}
\end{bmatrix}$$
we have the matrix-form system model
\begin{eqnarray}\label{eq:noise}
y+Ya=u+e
\end{eqnarray}

Here Lasso is used to solve the problem:
\begin{eqnarray}\label{eq:blind_deco}
\min_{u\in\mathbb{R}^n,a\in\mathbb{R}^p}\frac{1}{2}\|y+Ya-u\|_2^2+\lambda\|u\|_1
\end{eqnarray}

We can show that the solution of Lasso is very close to the true
$a^*$ and $u^*$. Before stating the theorem, we first introduce some
notation and technical conditions that will be used in the proof.

We denote the noiseless version of $Y$ as
$$X = \begin{bmatrix}
 0 & \cdots & 0\\
 x_0 & \cdots &0\\
  \vdots & \ddots & \vdots\\
x_{p-1} &\cdots & x_0\\
 \vdots & \vdots & \vdots\\
 x _{n-2} & \cdots & x_{n-p}
\end{bmatrix}$$
Denote the support of $u^*$ as $I$. We define $X_1$ as the matrix
comprising of the rows of $X$ indexed by $I$ and $X_2$ as the matrix
comprising of the rows of $X$ indexed by $I^c$. We also denote
$x_{\max}= \max_i |x_i|,\, u_{\min}=\min_{i\in I} |u_i|$ and
$a_{\max}=\max_{i} |a_i|$.

We assume that the AR process $x(n)$ satisfies the following set of
conditions.

\begin{description}
  \item[(1)] The smallest eigenvalue $\lambda_{min}(X_2^TX_2)\geq \frac{\|x\|_2^2}{c}\geq \frac{4np\sigma^2}{(\sqrt{2}-1)^2}$
  for some constant $c>1$.
  \item[(2)] $\|X_1^T\text{sgn}(z_I^*)\|_\infty\leq \|x\|_2\sqrt{\log n}$,
  \item[(3)] $x_{\max}\geq 2\sigma \sqrt{\log n}$  and
  $\frac{x_{\max}^2}{\|x\|_2^2}\leq \min\{\frac{1}{4c\sqrt{2pn}},\left(\frac{1}{24cp\sqrt{\log n}}\right)^2\}$.
\end{description}

In practice, condition (1) is generally satisfied. For instance, if the signal $x$ is persistent, $\frac{1}{\|x\|_2}X_2^TX_2$ converges to a constant invertible matrix. Condition (3) is also standard in compressed sensing, which says we need $SNR\geq O(\log
n)$. In addition, we also need the assumption that no components are
dominantly large (compared with the total energy of $x$). The upper
bound for $x_{\max}/\|x\|_2$ can be relaxed but the current setup
simplifies the analysis.

Condition (2) is new. Let us consider two scenarios. In the first
scenario, each spike in $u_{I}$ can be either positive or negative
with equal probability (i.e. $\text{sgn}(u_{I})$ is Bernoulli $\pm
1$). In this case, $X_1^T\text{sgn}(u_I^*)$ behaves like a
sub-Gaussian sum and it is usually upper bounded by
$\|x\|_2\sqrt{\log n}$ with high probability. On the other hand, let
us also consider the case when all the spikes in $u_{I}$ are of the
same sign, say {\em positive}. In this case each entry in $X_1^T$
and $\text{sgn}(u_I^*)$ is positive and the inner product of these
two aligned signals is typically much larger than the first
scenario. This phenomena is also illustrated in the experiments
shown in Figure \ref{fig:sgn}. In the experiment, the AR model is
$x_t-1.4x_{t-1}+0.45x_{t-2}=u(t)$. The blue curve corresponds to the
scenario when $\text{sgn}(u_i)$ ($u_i$ is a spike) is Bernoulli $\pm
1$. The red curve corresponds to the case when the sign of any spike
$u_i$ is always $+1$. Each point on the curve is an average over 40
trials. We can see that in the first scenario (blue curve) we can
tolerate many more spikes. To the best of our knowledge, this
behavior does not exist in standard compressed sensing problem.

\begin{figure}[t]
\begin{centering}
\includegraphics[width = .55\textwidth]{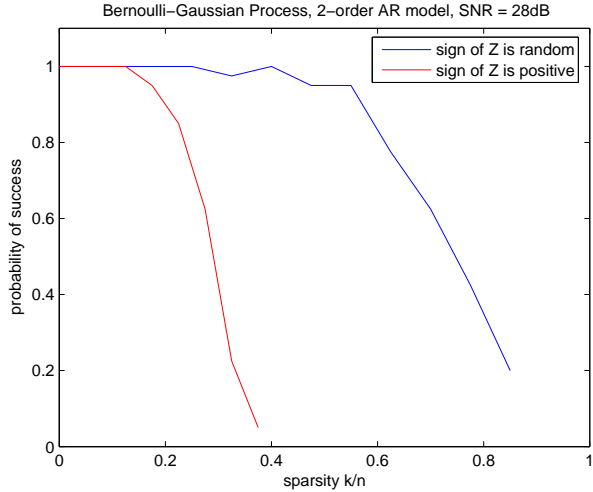}
\caption{\small Comparison of two sign conditions for $u$. The AR
model is $x_t-1.4x_{t-1}+0.45x_{t-2}=u(t)$. Noises ere added to the
measurements and $SNR=28$ dB. In one experiment, each sign of each
spike is either positive or negative with equal probability. In the
other experiment, the sign of the spikes is always positive. }
\label{fig:sgn}
\end{centering}
\end{figure}

\begin{thm}\label{thm2}
Denote $P:=I-Y(Y^TY)^{-1}Y^T$ and assume condition (1),(2) and (3)
stated above are satisfied. We also assume parameter $\lambda$ is
chosen such that $\lambda\geq 6\sigma pa_{\max}\sqrt{\log n}$ and
$u_{\min}\geq 2\lambda$, the solution to Lasso \ref{eq:blind_deco}
is given by
\begin{eqnarray}
\hat{u}_I &=& (P^T_IP_I)^{-1}(P^T_Ie-\lambda\text{sgn}(u^*_I))+u^*_I\\
\hat{u}_{I^c} &=& 0\\
 \hat a &=& -(Y^TY)^{-1}Y^T(y-\hat u)
\end{eqnarray}
and we have $\text{sgn}(u^*)=\text{sgn}(\hat u)$ with probability at
least $1-8p/n-(p+1)2^{-n/5}$.
\end{thm}
{\noindent \bf Remark: } The assumption $u_{\min}\geq 2\lambda$
implicitly implies an SNR bound $O(\log n)$ for the smallest spike. The
assumption $\lambda\geq 6\sigma pa_{\max}\sqrt{\log n}$ ensures
$\lambda$ to be sufficiently large so that every non-spike element
is shrunk to zero by the Lasso estimator. It is hard to analyze the
case when parameter $\lambda$ is smaller because in this case it is not clear how to construct $\hat u_{I^c}$ which is critical for tractable KKT analysis. The choice of $\hat u$ in the
Theorem \ref{thm2} is motivated by the proof techniques used in \cite{Candes4}. The proof of Theorem \ref{thm2} is presented in Section \ref{sec:proof2}.

\section{Extensions}\label{sec:ext}
In this section, we provide two interesting extensions to
the AR model problem. First, we generalize AR model to the
autoregressive moving average (ARMA) model, i.e., the process
contains both poles and zeros in the transform function. Second, we
develop an algorithm for the non-causal AR process, i.e.,
the current state not only depends on the past inputs but also
depends on the future inputs.

\subsection{ARMA model}
The ARMA model takes the form
\begin{eqnarray}
x(n)+\sum_{i=1}^p a_i x(n-i)= u(n)+\sum_{i=1}^q b_i u(n-i)
\end{eqnarray}
Again we use Equation \ref{eq:sense} to obtain the measurement $y=Gx$
where $G$ is a Toeplitz matrix as defined in Section
\ref{sec:setup}. Similar to what we have done in Section
\ref{sec:setup}, we write down the matrix representation of the ARMA
model:
\begin{eqnarray}\label{eq:ARMA_basic}
\begin{bmatrix}
x_0 & 0 & \cdots & 0\\
x_1 & x_0 & \cdots &0\\
\vdots & \vdots & \ddots & \vdots\\
x_{p} &x_{p-1} &\cdots & x_0\\
\vdots & \vdots & \vdots & \vdots\\
x_{n-1} & x _{n-2} & \cdots & x_{n-p}
\end{bmatrix}
\begin{bmatrix}
1\\a_1\\\vdots\\a_{p}
\end{bmatrix}
=\begin{bmatrix}
 1 & 0 & 0 & \cdots & 0 \\
b_1 & 1 & 0 & \cdots & 0\\
\vdots & \vdots & \ddots & \ddots & \vdots\\
b_q & \cdots & b_1 & 1& \cdots\\
\vdots & \vdots & \ddots & \ddots & \vdots\\
0 & \cdots & b_{q} & \cdots b_1 & 1
\end{bmatrix}
\begin{bmatrix}
u_0\\
u_1 \\
\vdots\\
u_{n-2}\\ u_{n-1}
\end{bmatrix}
\end{eqnarray}
We denote the lower triangular matrix $B$ as
\begin{eqnarray}\label{eq:S}
B := \begin{bmatrix}
 1 & 0 & 0 & \cdots & 0 \\
b_1 & 1 & 0 & \cdots & 0\\
\vdots & \vdots & \ddots & \ddots & \vdots\\
b_q & \cdots & b_1 & 1& \cdots\\
\vdots & \vdots & \ddots & \ddots & \vdots\\
0 & \cdots & b_{q} & \cdots b_1 & 1
\end{bmatrix}
\end{eqnarray}
By multiplying $G^{[m-p]}$ to both sides of Equation
\ref{eq:ARMA_basic}, we get
\begin{eqnarray}\label{eq:model1}
Ya+y^{[m-p]} = G^{[m-p]}Bu
\end{eqnarray}
Note that for ARMA model we have an additional term $B$ compared to
Equation \ref{eq:model}. Generally, matrix $B$ is unknown. We first
consider a simple situation when $B$ is assumed to be known to the
decoder. Based on Theorem \ref{thm1} we can derive the following result.
.
\begin{thm}[Known Zero Locations] \label{thm:knownzeros}
Given the same technical conditions as Theorem \ref{thm1} and assume
$u^*$ is the original sparse spike train that generates the ARMA
process. Then $u^*$ is the unique minimizer of
\begin{eqnarray}\label{eq:ARMA_cor}
\min_{u\in\mathbb{R}^n,a\in\mathbb{R}^p}\quad \|u\|_1 \,\,\text{
subject to}\quad Ya+y^{[m-p]}=G^{[m-p]}Bu
\end{eqnarray}
\end{thm}
\begin{proof}
Note that $B$ is also a Toeplitz matrix. From the commutativity of
Toeplitz matrix, we have $G^{[m-p]}B=BG^{[m-p]}$.  From Section
\ref{sec:proof}, the KKT conditions claim that $u^*$ is the unique
minimizer of Equation \ref{eq:ARMA_cor} if and only if there exists
a vector $\pi$ such that:
\begin{enumerate}
  \item $\left(\pi^TG^{[m-p]}B\right)_i = \text{sgn}(u^*_i)$ for all
  $i\in\text{Supp}(u^*)$,
  \item $|\left(\pi^TG^{[m-p]}B\right)_j|<1$ for all
  $j\not\in\text{Supp}(u^*)$,
  \item $\pi^TY=0$.
\end{enumerate}

Applying the commutativity and define $\tilde\pi^T=\pi^T B$, the
above three conditions are converted to
\begin{enumerate}
  \item $\left(\tilde\pi^TG^{[m-p]}\right)_i = \text{sgn}(u^*_i)$ for all
  $i\in\text{Supp}(u^*)$,
  \item $|\left(\tilde\pi^TG^{[m-p]}\right)_j|<1$ for all
  $j\not\in\text{Supp}(u^*)$,
  \item $\tilde\pi^T B^{-1}Y=0$.
\end{enumerate}
Note that both the inverse $B^{-1}$ and the matrix $Y$ are
Toeplitz. Therefore, from commutativity, the third equation is
equivalent to $\tilde\pi^TYB^{-1}=0$. Finally, since $B^{-1}$ is
invertible, the last equation can be further simplified to
$\tilde\pi^T Y=0$. Now the KKT conditions look exactly the same as
those in Section \ref{sec:proof}. Hence the corollary is proved by
following the same argument as in Section \ref{sec:proof}.
\end{proof}

Now we consider the general situation when $B$ is unknown. The
difficulty of decoding lies in the fact that we know neither $B$ nor
the spike train $u(n)$. There might exist different combinations of
$b_i$ and $u(n)$ that matches the measurements $y(n)$.

Here we propose an iterative algorithm for estimating $(u,a,b)$ in Equation~\ref{eq:model1}. Each iteration comprises of two basic
steps. First, if $B$ is known (from previous iteration), we can use
the following $\ell_1$ minimization algorithm to solve $u$ and $a$
(Theorem~\ref{thm:knownzeros}).
\begin{eqnarray}\label{eq:ARMA_lasso}
\min_{u\in \mathbb{R}^n,a\in \mathbb{R}^p}\|u\|_1 \quad \text{s.t.
}\quad\|Ya+y^{[m-p]} - G^{[m-p]}Bu\|_2\leq \epsilon
\end{eqnarray}
Here $\epsilon>0$ is required, even though there may not be any noise, to ensure that we do not get stuck in a local minima.

Now once $u$ is determined we switch from $u$ to $B$, as the optimization variable. This problem reduces to a standard regression problem. First we rewrite Equation~\ref{eq:model1} as follows:
\[Ya+y^{[m-p]} = G^{[m-p]}\begin{bmatrix}u_0 & 0 & \cdots & 0\\
u_1 & u_0& \cdots & 0\\
\vdots & \vdots & \ddots & \vdots\\
u_{q} & u_{q-1} &\cdots & u_0\\
\vdots & \vdots & \ddots & \vdots\\
u_{n-1} & u_{n-2} & \cdots & u_{n-q-1}
\end{bmatrix}
\begin{bmatrix}
1 \\ b_1 \\ b_2 \\ \vdots \\ b_q\end{bmatrix}
\]
which can be simplified to $Ya+y^{[m-p]} = G^{[m-p]}u+G^{[m-p]}Ub$ where we denote $$U = \begin{bmatrix} 0 & \cdots & 0\\
u_0& \cdots & 0\\
 \vdots & \ddots & \vdots\\
 u_{q-1} &\cdots & u_0\\
 \vdots & \ddots & \vdots\\
 u_{n-2} & \cdots & u_{n-q-1}
\end{bmatrix}$$

Now we formulate the following least squares optimization problem:
\begin{eqnarray}\label{eq:ls}
\min_{b\in\mathbb{R}^q}\|Ya+y^{[m-p]}-G^{[m-p]}u-G^{[m-p]}Ub\|_2\end{eqnarray}

\begin{figure}[t]
\begin{centering}
\begin{minipage}[t]{.48\textwidth}
\includegraphics[width = 1\textwidth]{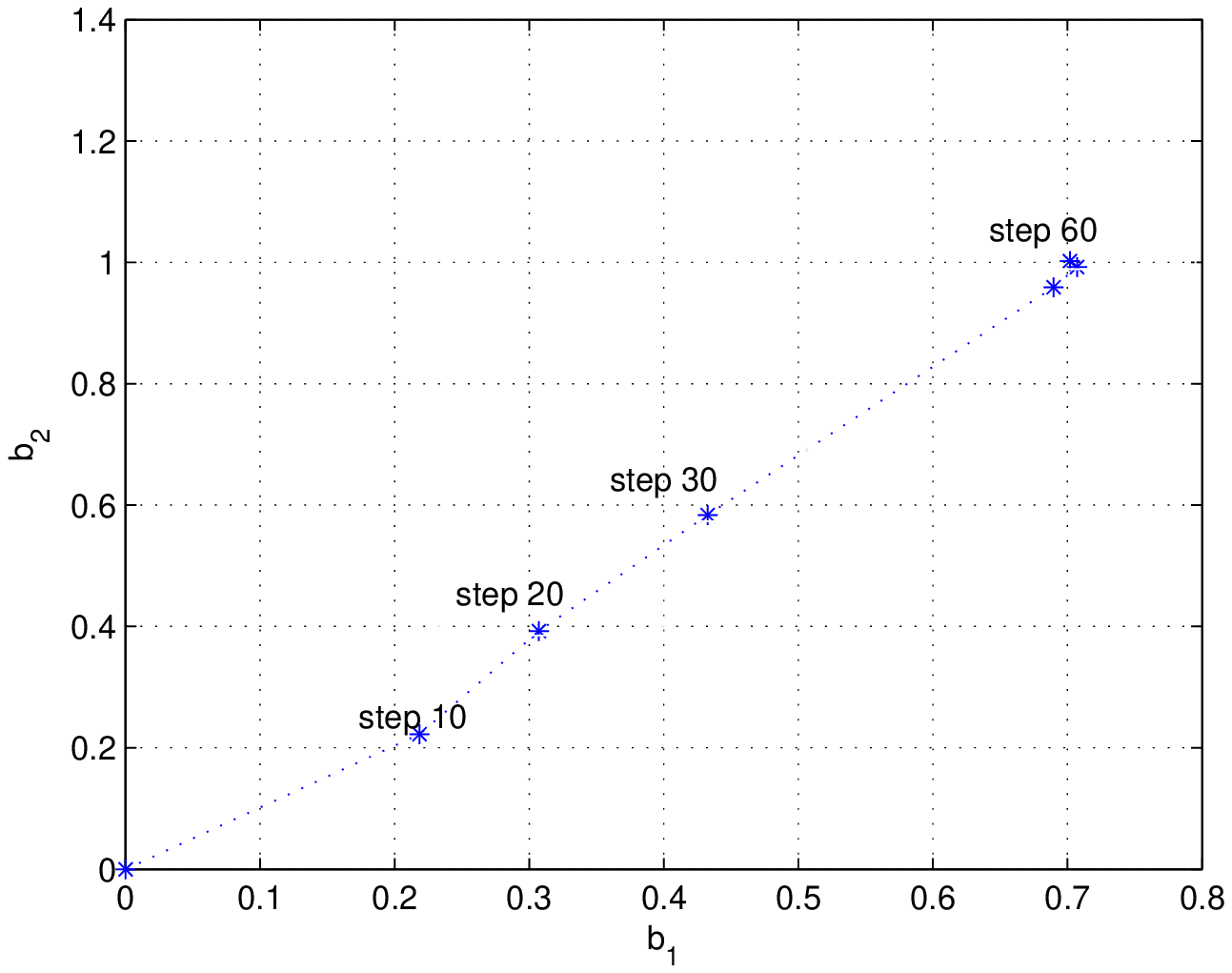}\\
\end{minipage}
\begin{minipage}[t]{.48\textwidth}
\includegraphics[width = 1\textwidth]{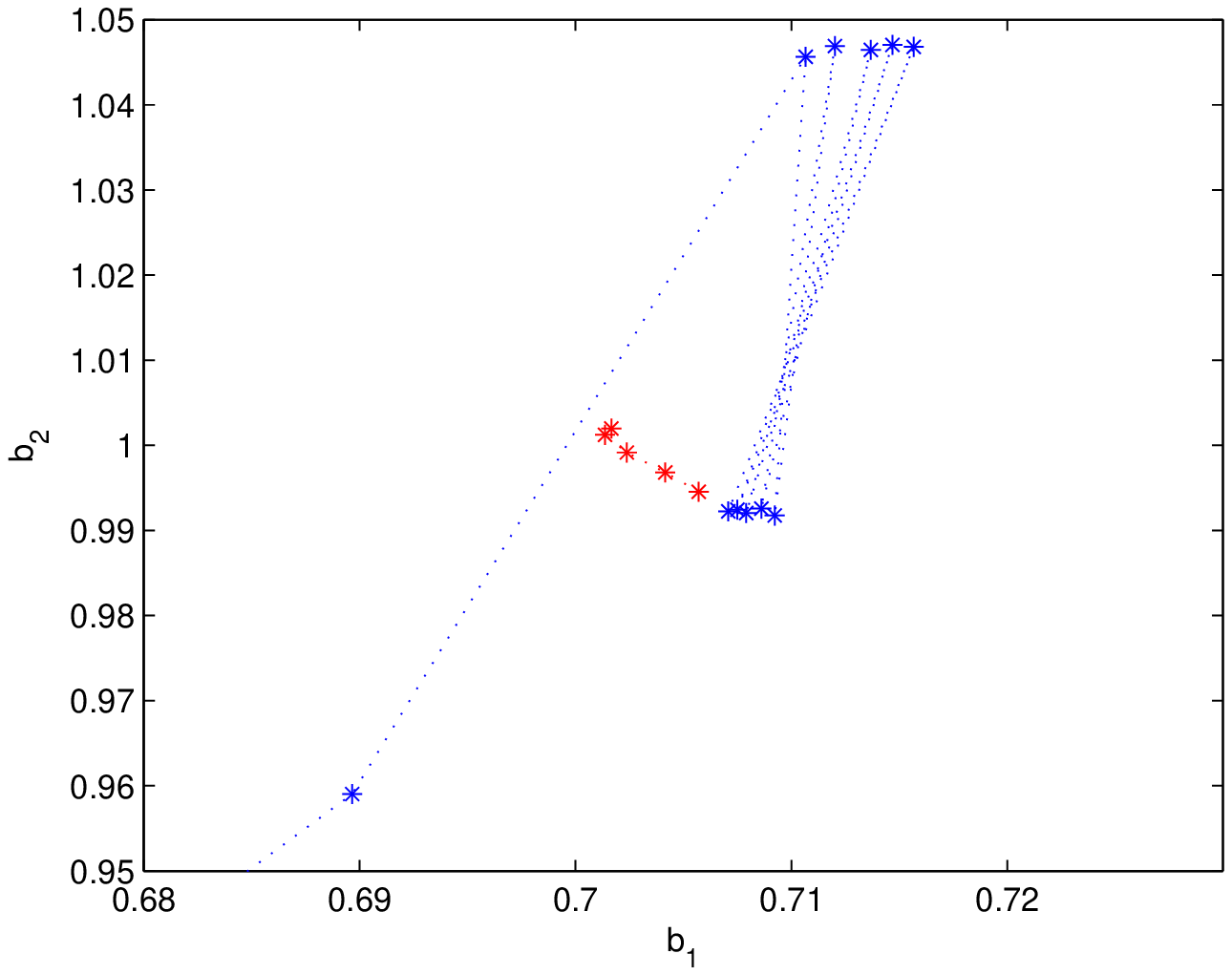}\\
\end{minipage}
\caption{\small The iterative algorithm on the model
$x(n)-1.9x(n-1)+1.06x(n-2)-0.144x(n-3)=u(n)+0.7u(n-1)+u(n-2) $ where
the correct $b=[0.7\,\,1]$. {\em Left: } In trajectory of $\hat b$
in each round of iteration; {\em Right: } Zoom-in of the final
stages of the iterations. Blue $*$ corresponds to the rounds of
updates with $\epsilon=3$ while red $*$  corresponds to the rounds
of updates with a smaller $\epsilon=0.3$ in the final stage.}
\label{fig:locus}
\end{centering}
\end{figure}

In summary our iterative algorithm consists of the following steps:
\begin{description}
  \item[Initialization: ] Set $b^{(0)}=0$, i.e., $B^{(0)} = I$.
  \item[Iteration $k$: ] Compute $u^{(k)},a^{(k)},b^{(k)}$
\begin{enumerate}
  \item Update $u^{(k)}$ and $a^{(k)}$ via solving Equation \ref{eq:ARMA_lasso} with $B=B^{(k-1)}$;
  \item Update $b^{(k)}$ via solving least-square (Equation \ref{eq:ls}) with $(u,a)=(u^{(k)},a^{(k)})$.
\end{enumerate}
\end{description}

There is a subtlety in the choice of parameter $\epsilon$ in Equation
\ref{eq:ARMA_lasso}. If $\epsilon$ is large, the iterative algorithm
appears to have a faster convergence rate but at the cost of significant bias. On the other hand, if $\epsilon$ is
small, the convergence rate is slow but the solution has small bias.  Therefore, in practical implementation we
choose $\epsilon$ to be reasonably large in the early stages of the
iteration and then decrease it to $\epsilon/10$ at the later stages
of the iteration.

Figure \ref{fig:locus} illustrates a concrete example of solving the
ARMA model
$x(n)-1.9x(n-1)+1.06x(n-2)-0.144x(n-3)=u(n)+0.7u(n-1)+u(n-2)$ by
using our iterative algorithm. We choose $\epsilon=3$ in the first
50 rounds of iteration and finally in the last 10 rounds of updates
we set $\epsilon = 0.3$.  Figure \ref{fig:locus}(b) is a zoom-in
version of Figure \ref{fig:locus}(a) which shows the final stage of
the algorithm. We can see the effects of choosing different value of
$\epsilon$ as well.

\subsection{Non-causal AR model}
Many real world signals are non-causal. For example, a 2D image is
usually modeled by a Markov random field, where each pixel is
dependent on all its neighboring pixels. In this subsection we consider this situation by modeling the signal to be a non-causal
AR process.

A non-causal AR model is defined as
\begin{eqnarray}\label{eq:noncausal}
x(n)+\sum_{i=1}^pa_i x(n-i)+\sum_{i=1}^p a_{-i}x(n+i)=u(n)
\end{eqnarray}

\begin{figure}[t]
\begin{centering}
\includegraphics[width = .55\textwidth]{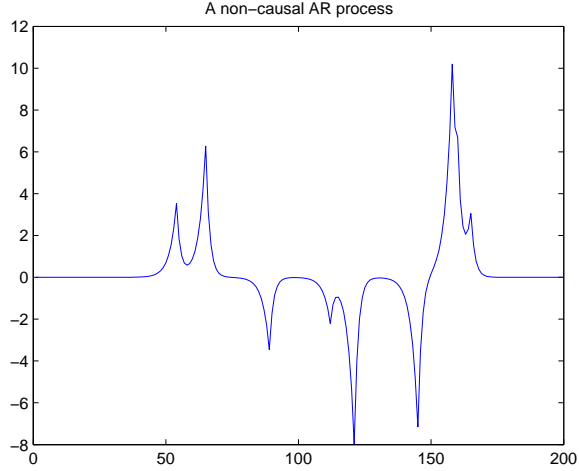}
\caption{\small A typical non-causal Autoregressive process:
$x(n)-0.375x(n-1)-0.5x(n+1)=u(n)$.} \label{fig:noncausal}
\end{centering}
\end{figure}

A typical non-causal AR process is shown in Figure
\ref{fig:noncausal}. Here the impulse response of each spike is
two-sided as opposed to the one-sided impulse response of causal AR
process. In this subsection, we discriminate between two boundary
conditions for the non-causal AR process. As we will show later,
there are subtle differences in dealing with these two boundary
conditions.

\begin{enumerate}
  \item Boundary is circulant, i.e., $x_0=x_n,x_1=x_{n+1},\cdots$;
  \item Boundary is not circulant.
\end{enumerate}

\subsubsection{Circulant Boundary}
In this case we use the following circulant matrix in the sensing
model $y=Gx$.
\begin{eqnarray}
G =
\begin{bmatrix}
g_{n-m}& g_{n-m-1}&\cdots & g_0 & g_{n-1} & g_{n-2}&\cdots & g_{n-m+1}\\
g_{n-m+1}& g_{n-m} & \cdots & g_1 & g_0  & g_{n-1} & \cdots &
g_{n-m}\\
\vdots & \vdots & \ddots & \vdots & \vdots & \ddots &\ddots &
\vdots\\
g_{n-1} & g_{n-2} & \cdots& g_{m-1} & g_{m-2} & g_{m-3}&\cdots  &
g_0
\end{bmatrix}
\in \mathbb{R}^{n\times m}
\end{eqnarray}
where $g_i$ is i.i.d Gaussian random variable $\mathcal{N}(0,1)$ or
Bernoulli $\pm 1$ random variable.

Since the boundary of $x$ is circulant ($x_{-i}=x_{n-i}$), we can
write the matrix representation of Equation \ref{eq:noncausal} as
\begin{eqnarray}\label{eq:nonmat}
\begin{bmatrix}
x_0 & x_{n-1} & \cdots & x_{n-p} & x_{1} & \cdots & x_{p}\\
x_1 & x_0 & \cdots & x_{n-p+1} & x_{2} & \cdots & x_{p+1}\\
\vdots & \vdots & \ddots & \vdots & \vdots & \ddots & \vdots\\
x_{p} &x_{p-1} &\cdots & x_0 &x_{p+1} &\cdots & x_{2p}\\
\vdots & \vdots & \vdots & \vdots& \vdots & \vdots & \vdots\\
x_{n-1} & x _{n-2} & \cdots & x_{n-p}& x _{0} & \cdots & x_{p-1}
\end{bmatrix}
\begin{bmatrix}
1\\a_1\\\vdots\\a_{p}\\a_{-1}\\\vdots\\a_{-p}
\end{bmatrix}
=\begin{bmatrix}
u_0\\
u_1 \\
\vdots\\
u_{p}\\
\vdots\\
u_{n-1}
\end{bmatrix}
\end{eqnarray}

With an abuse of notation, we use $G^{[i:j]}$ to denote
the submatrix of $G$ comprising rows $i$-th through $j$-th
of $G$. Now we multiply $G^{[p+1:m-p]}$ to both sides of
Equation \ref{eq:nonmat} we get the following equation.
\begin{eqnarray}\label{eq:nonbasic}
\begin{bmatrix}
y_{p} & y_{p-1} & \cdots & y_{0} & y_{p+1} & \cdots & y_{2p}\\
y_{p+1} & y_{p} & \cdots & y_{1} & y_{p+2} & \cdots & y_{2p+1}\\
\vdots & \vdots & \ddots & \vdots & \vdots & \ddots & \vdots\\
y_{2p} &y_{2p-1} &\cdots & y_{p} & y_{2p+1} &\cdots & y_{3p}\\
\vdots & \vdots & \vdots & \vdots& \vdots & \vdots & \vdots\\
y_{n-p} & y _{n-p-1} & \cdots & y_{n-2p}& y _{n-p+1} & \cdots &
y_{n}
\end{bmatrix}
\begin{bmatrix}
1\\a_1\\\vdots\\a_{p}\\a_{-1}\\\vdots\\a_{-p}
\end{bmatrix}
=G^{[p+1:m-p]}u
\end{eqnarray}
We define matrix $\tilde Y$ to be
\begin{eqnarray*}
\tilde Y=
\begin{bmatrix}
 y_{p-1} & \cdots & y_{0} & y_{p+1} & \cdots & y_{2p}\\
 y_{p} & \cdots & y_{1} & y_{p+2} & \cdots & y_{2p+1}\\
\vdots & \ddots & \vdots & \vdots & \ddots & \vdots\\
y_{2p-1} &\cdots & y_{p} & y_{2p+1} &\cdots & y_{3p}\\
\vdots & \vdots & \vdots& \vdots & \vdots & \vdots\\
 y _{n-p-1} & \cdots & y_{n-2p}& y _{n-p+1} & \cdots &
y_{n}
\end{bmatrix}\end{eqnarray*}
and finally Equation \ref{eq:nonbasic} is simplified to
\begin{eqnarray}
y^{[p+1:m-p]}+\tilde Y a = G^{[p+1:m-p]}u
\end{eqnarray}
where $a = [a_1,\cdots,a_p,a_{-1},\cdots,a_{-p}]^T\in
\mathbb{R}^{2p}$.

As in Section~\ref{sec:alg} we can use either $\ell_1$-minimization or Lasso to solve this problem.
\begin{eqnarray*}
\ell_1\text{-minimization:}\quad
\min_{u\in\mathbb{R}^n,a\in\mathbb{R}^{2p}}\|u\|_1\quad
\text{s.t.}\quad y^{[p+1:m-p]}+\tilde Y a = G^{[p+1:m-p]}u\\
\text{Lasso:}\quad
\min_{u\in\mathbb{R}^n,a\in\mathbb{R}^{2p}}\dfrac{1}{2}\|y^{[p+1:m-p]}+\tilde
Y a - G^{[p+1:m-p]}u\|_2^2+\lambda\|u\|_1
\end{eqnarray*}

\subsubsection{Non-circulant Boundary}

The case of non-circulant boundary is slightly more complicated.
There are two ways of handling this situation. A simple approach is to
view the problem as a perturbation of the circulant boundary case, namely,
\[y^{[p+1:m-p]}+\tilde Y a +e= G^{[p+1:m-p]}u\]
where
\begin{eqnarray*}
e = G^{[p+1:m-p]}\begin{bmatrix}
x_{-1}-x_{n-1} & \cdots & x_{-p}-x_{n-p} & 0 & \cdots &0\\
 0 & \cdots & x_{-p+1}-x_{n-p+1} & 0 & \cdots & 0\\
 \vdots & \ddots & \vdots & \vdots & \ddots & \vdots\\
0 &\cdots & 0 & 0 &\cdots & x_{n+p-2}-x_{p-2}\\
 0 & \cdots & 0& x_n-x _{0} & \cdots & x_{n+p-1}-x_{p-1}
\end{bmatrix}
a
\end{eqnarray*}

Now one could use Lasso to solve this noisy model:
\begin{eqnarray*}
\min_{u\in\mathbb{R}^n,a\in\mathbb{R}^{2p}}\dfrac{1}{2}\|y^{[p+1:m-p]}+\tilde
Y a - G^{[p+1:m-p]}u\|_2^2+\lambda\|u\|_1
\end{eqnarray*}

Unfortunately, this approach will have a bias. To overcome this
limitation, we consider the case where we can make an additional $2p$ set of
measurements corresponding to the boundary conditions, namely,
$$y_{m+1}=x_{-p},\,\cdots,\,y_{m+p}=x_{-1},\,y_{m+p+1}=x_{n-p},\,\cdots,\,y_{m+2p}=x_{n-1}.$$
Then by the denoting $$\bar Y := \tilde
Y+G^{[p+1:m-p]}\begin{bmatrix}
x_{-1}-x_{n-1} & \cdots & x_{-p}-x_{n-p} & 0 & \cdots &0\\
 0 & \cdots & x_{-p+1}-x_{n-p+1} & 0 & \cdots & 0\\
 \vdots & \ddots & \vdots & \vdots & \ddots & \vdots\\
0 &\cdots & 0 & 0 &\cdots & x_{n+p-2}-x_{p-2}\\
 0 & \cdots & 0& x_n-x _{0} & \cdots & x_{n+p-1}-x_{p-1}
\end{bmatrix}$$
the sensing model can be simplified to the noiseless version
\[y^{[p+1:m-p]}+\bar Y a= G^{[p+1:m-p]}u\]

Again we can use either $\ell_1$-minimization or Lasso to solve this
model:
\begin{eqnarray*}
\ell_1\text{-minimization:}\quad
\min_{u\in\mathbb{R}^n,a\in\mathbb{R}^{2p}}\|u\|_1\quad
\text{s.t.}\quad y^{[p+1:m-p]}+\bar Y a = G^{[p+1:m-p]}u\\
\text{Lasso:}\quad
\min_{u\in\mathbb{R}^n,a\in\mathbb{R}^{2p}}\dfrac{1}{2}\|y^{[p+1:m-p]}+\bar
Y a - G^{[p+1:m-p]}u\|_2^2+\lambda\|u\|_1
\end{eqnarray*}

\section{Proof of Theorem \ref{thm1}}\label{sec:proof}

We first write down the primal and dual formulation of algorithm
\ref{eq:form1}.
\begin{eqnarray}\label{eq:primal}
\min_{u\in\mathbb{R}^n,a\in\mathbb{R}^p}\quad \|u\|_1 \,\,\text{
subject to}\quad Ya+y^{[m-p]}=G^{[m-p]}u
\end{eqnarray}
whose \emph{dual} formualtion is:
\begin{eqnarray}\label{eq:dual}
\max_{\pi\in\mathbb{R}^m}\quad \pi^Ty^{[m-p]}\,\text{ subject
to}\quad \|\pi^TG^{[m-p]}\|_\infty\leq 1,\,\pi^TY=0
\end{eqnarray}

The proof is based on duality. $u^*$ is the unique minimizer of the
primal problem \ref{eq:primal} if we can find a dual vector $\pi$
with the following properties:
\begin{enumerate}
  \item $\left(\pi^TG^{[m-p]}\right)_i = \text{sgn}(u^*_i)$ for all
  $i\in\text{Supp}(u^*)$,
  \item $|\left(\pi^TG^{[m-p]}\right)_j|<1$ for all
  $j\not\in\text{Supp}(u^*)$,
  \item $\pi^TY=0$.
\end{enumerate}
where $\text{sgn}(u^*_i)$ denotes the sign of $u^*_i$
($\text{sgn}(u^*_i)=0$ for $u^*_i=0$) and $\text{Supp}(u^*)$ denotes
the support of vector $u^*$. The above set of conditions ensure that
the primal-dual pair $(u^*,\pi)$ is not only feasible but also
satisfy the complementary slackness condition, thus optimal. We call
the above three conditions as the \emph{Dual Optimal Condition}
(DOC).

The rest of this section is to construct a $\pi$ that satisfies the
DOC. Our construction relies on the following result (see
\cite{Candes5}).
\begin{lem}[\cite{Candes5}]\label{lem:free}
Let $S\geq 1$ be such that $\delta_{2S}\leq\frac{1}{3}$, and $c$ be
a real vector supported on $T$ obeying $|T|\leq S$. Then there
exists a vector $\pi\in \mathbb{R}^{m}$ such that
$\left(\pi^TG^{[m-p]}\right)_i = c_i\,\forall i\in T$. Furthermore,
$\pi$ obeys
\[\left|\left(\pi^TG^{[m-p]}\right)_j\right|\leq \frac{\delta_S}{(1-3\delta_{2S})\sqrt{S}}\cdot \|c\|_2\quad\forall j\not\in T\]
\end{lem}

This lemma gives us the freedom to choose (arbitrarily) the value of
$\pi^TG^{[m-p]}$ in the location of $T$ while the magnitude of the
rest components  is still bounded.

\subsection{One Pole Case}
In this section we provide a proof for the simple case when $x(n)$
is a first order AR process (i.e., $p=1$) and $u^*$ only contains
one spike (i.e., every entry of $u^*$ is zero except one place).
Though simple, it contains the main idea of proof techniques for the
more general case. Note that in this simple case the assumptions in
Theorem \ref{thm1} are automatically satisfied.

For the $1$-sparse driving process $u^*$, without loss of any
generality we assume $u^*_0=1$ and $u^*_i=0\,(\forall i\geq 1)$. We
also denote $\alpha = -a$ as the root of the characteristic function
of the first order AR process. Due to stability we have
$|\alpha|<1$. Now in condition 3 of DOC, the term $\pi^T Y$ can be
recast as
\[\pi^T Y =
\pi^TG^{[m-p]}\begin{bmatrix}0\\x_0\\\vdots\\x_{n-2}\end{bmatrix}=\pi^TG^{[m-p]}\begin{bmatrix}0\\1
\\ \vdots\\\alpha^{n-2}\end{bmatrix}\]

In Lemma \ref{lem:free}, we choose $c$ as $c_0=1$, $c_1=1/2$ and
$c_j=0$ ($j =2,\cdots,S-1$). Then Lemma \ref{lem:free} tells us that
there exists a $\pi_1$ such that $\left(\pi_1^TG^{[m-p]}\right)_i =
c_i\,(\forall i=0,\cdots, S-1)$ and furthermore
\[\left|\left(\pi_1^TG^{[m-p]}\right)_j\right|\leq \frac{\delta_S}{(1-3\delta_{2S})\sqrt{S}}\cdot \sqrt{1+1/4}\leq \frac{2}{\sqrt{S}},\quad\forall j\geq S\]

This implies \begin{eqnarray}\pi_1^T Y
=\pi_1^TG^{[m-p]}\begin{bmatrix}0\\1
\\ \vdots\\\alpha^{n-2}\end{bmatrix}=\frac{1}{2}+\sum_{j=S}^{n-1}\left(\pi_1^TG^{[m-p]}\right)_j\alpha^{j-1}\end{eqnarray}
where the summation
$\left|\sum_{j=S}^{n-1}\left(\pi_1^TG^{[m-p]}\right)_j\alpha^{j-1}\right|\leq
\frac{2|\alpha|^{S-1}}{\sqrt{S}(1-|\alpha|)}\ll\frac{1}{2}$.
Therefore $\text{sgn}(\pi_1^T Y)=1$. To summarize the above
discussion, we find $\pi_1$ such that:
\begin{enumerate}
  \item $\left(\pi_1^TG^{[m-p]}\right)_0 = 1$
  \item $|\left(\pi_1^TG^{[m-p]}\right)_j|<1$ for all
  $j\geq1$,
  \item $\text{sgn}(\pi_1^TY)=1$.
\end{enumerate}

Similarly, by choosing $c_0=1$, $c_1=-1/2$ and $c_j=0$ ($j
=2,\cdots,S-1$) in Lemma \ref{lem:free}, there exists a $\pi_2$ such
that condition 1 and 2 of DOC are also satisfied while
$\text{sgn}(\pi_2^T Y)=-1$. Hence, by convexity there exists a
$\lambda\in(0,1)$ such that for $\pi =
\lambda\pi_1+(1-\lambda)\pi_2$, it satisfies $\pi^TY=0$ and also
condition 1 and 2, i.e., the whole DOC.

Finally we find a primal-dual pair $(u^*,\pi)$ that satisfy all the
feasible constraints and also the complementary slackness condition,
which implies $u^*$ is the unique minimizer of the primal problem
equation \ref{eq:primal}.

\subsection{General Case}\label{sec:general}
In this section we prove that in the general case the three
conditions in Theorem \ref{thm1} ensures the existence of a $\pi$
that satisfies the DOC. Before giving the proof, we point out that
if some conditions in Theorem \ref{thm1} are violated, there might
not exist such a $\pi$. Let us consider the case of $p=1$ (first
order AR process) and $k$=2 (only two entries of $u(n)$ are
nonzero). Moreover, we choose $u^*_0=u^*_1=1$ and $u^*_i=0\,(\forall
i>1)$, that is, the two spikes are next to each other.

In this case
$[x_0,x_1,\cdots,x_{n-1}]^T=[1,1+\alpha,\alpha(1+\alpha),\cdots,\alpha^{n-2}(1+\alpha)]^T$.
We pick $\alpha = -1/2$. Clearly the assumption $|i-i'|>l,\forall
i,i'\in\text{Supp}(u^*)$ in Theorem \ref{thm1} is broken. On the
other hand, we can also check that there does not exist a $\pi$ that
satisfies the whole DOC condition. In fact, suppose $\pi$ is chosen
such that condition 1 and 2 are satisfied, then in checking
condition 3 we find

\begin{eqnarray*}\pi^T Y
=\pi^TG^{[m-p]}\begin{bmatrix}0\\1\\1+\alpha\\
\\ \vdots\\\alpha^{n-3}(1+\alpha)\end{bmatrix}=1+\sum_{j=2}^{n-1}\left(\pi_1^TG^{[m-p]}\right)_j\alpha^{j-2}(1+\alpha)
\geq 1-\sum_{j=2}^{n-1}|\alpha|^{j-2}(1-|\alpha|)>0\end{eqnarray*}
which violates condition 3 in DOC. Hence there does not exist a
$\pi$ that satisfies all the three conditions in DOC.

Before proving Theorem \ref{thm1}, we need the following lemma in
constructing $\pi$.
\begin{lem}\label{lem:cond}
Suppose the assumptions in Theorem \ref{thm1} are satisfied. Denote
$T = \{j+i:j\in \text{Supp}(u^*),0\leq i\leq l\}$. Then $|T|\leq S$
and we also have the following inequalities:
\begin{description}
  \item[(1)] $\forall j\not\in T$ and $i=0,1,\cdots,p$, $|x_{j-i}|< \frac{\beta_{\max}M\rho^{l-p}}{1-\rho^l}$,
  \item[(2)] $\forall k\in\cup_{i=1}^p\{j-i:j\in \text{Supp}(u^*)\}$, $|x_k|< \frac{\beta_{\max}M\rho^{l-p}}{1-\rho^l}$,
  \item[(3)] $\forall i=0,1,\cdots,p-1$ and $\forall j\in\text{Supp}(u^*), |x_{j}/x_{j+i}|\geq r$ where $r:=\frac{\beta_{\min}(1-\rho^l)}{\beta_{\max}M}-\rho^l$.
\end{description}
\end{lem}
\begin{proof}
First, from the assumption of Theorem \ref{thm1}, $|T|\leq S$. Then
we need to verify the three properties.

Suppose $u_k$ is a new spike and $k'$ be the next spike. Given
$i<(k'-k)$, we clearly have
\[|x_{k+i}|\leq \beta_{\max}M\rho^i(1+\rho^l+\rho^{2l}+\cdots)\leq \frac{\beta_{\max}M\rho^i}{1-\rho^l}\]

Hence properties (1) and (2) follow.

We denote $\epsilon := \frac{\beta_{\max}M}{1-\rho^l}$. Therefore,
for any $j\in\text{Supp}(u^*)$, $|x_j|>|u_j|-\epsilon\rho^l\geq
\beta_{min}-\epsilon\rho^l$. Combining with the above argument, we
have
\[|x_j/x_{j+i}|\geq (\beta_{\min}-\epsilon\rho^l)/\epsilon=\frac{\beta_{\min}(1-\rho^l)}{\beta_{\max}M}-\rho^l\]

Note that when $\rho^l\leq \frac{\beta_{\min}}{3\beta_{\max}M}$ as
given in the theorem assumption, we have
$r\geq\frac{\beta_{\min}}{3\beta_{\max}M}$.
\end{proof}
\noindent {\bf Remark:} Property (1) in Lemma \ref{lem:cond} says
that many components of $x(n)$ are small. Property (2) ensures that
before a new `spike' $u_j$ begins ($j\in\text{Supp}(u)$), the
amplitude of $x_{j-p},\cdots,x_{j-1}$ is already negligible (i.e.,
very close to zero) such that the new impulse response caused by
$u_j$ can be regarded as starting almost from zero level. Finally, property (3) says that when a new spike $u_j$ arrives,
the corresponding output $x_j$ is reasonably large compared to its neighbors.\\

Now we are ready to prove Theorem \ref{thm1}. Similar to the last
section's argument, the objective is to find a sequence of vectors
$\pi_1,\cdots,\pi_{2^p}$ such that any of $\pi_s (s=1,\cdots,2^p)$
satisfies the condition 1 and 2 of DOC while
\begin{eqnarray*}\label{eq:sgn}
\text{sgn}(\pi_1^TY)&=& [1,1,\cdots,1]^T\\
\text{sgn}(\pi_2^TY) &=& [-1,1,\cdots,1]^T\\
\vdots\\
\text{sgn}(\pi_{2^p}^TY)&=& [-1,-1,\cdots,-1]^T
\end{eqnarray*}
and this implies there exists a convex combination $\pi
=\sum_{s=1}^p \lambda_s \pi_s$ which satisfies $\pi^TY=0$ and also
the condition 1 and 2 of DOC.

Based on Lemma \ref{lem:free}, we construct $\pi_1$ via fixing the
values of $\{(\pi_1^TG^{[m-p]})_i\}_{i\in T}:=\{c_i\}_{i\in T}$:
\begin{eqnarray}\label{eq:pickup_c}
  c_i =
  \begin{cases}
  \text{sgn}(u^*_i)& \text{if }i\in
  \text{Supp}(u^*)\\
  (r/2)^{i-j-1}\text{sgn}(x_j) & \text{if }i=j+1,\cdots,j+p,\forall j\in\text{supp}(u^*)\\
  0 & \text{if }i=j+p+1,\cdots,j+l,\forall j\in\text{supp}(u^*)
 \end{cases}
\end{eqnarray}

This choice of $c$ gives the bound $\|c\|_2<
\sqrt{k+k(1+2^{-1}+2^{-2}+\cdots)}\leq \sqrt{3k}$. Now by applying
Lemma \ref{lem:free}, we know there exists a $\pi_1$ such that
$\left(\pi_1^TG^{[m-p]}\right)_i=c_i$ when $i\in T$ and
\begin{eqnarray*}
\left|\left(\pi_1^TG^{[m-p]}\right)_j\right|<
\frac{\|c\|_2}{\sqrt{S}}\leq\sqrt{\frac{3k}{S}}\leq1,\quad \forall
j\not\in T
\end{eqnarray*}
where the last inequality follows from the assumption of the
Theorem. Up to now we have shown that $\pi_1$ satisfies condition 1
and 2 of DOC. Next we will check the sign of $\pi_1^T Y$.

For $t = 1,2,\cdots,p$,
\begin{eqnarray*}
(\pi_1^T Y)_{t}
&=&\sum_{j=j_0+1}^{j_0+p}\left(\pi_1^TG^{[m-p]}\right)_j
x_{j-t}+\sum_{j\not\in T\text{ or
}j\in\text{Supp}(u^*)}\left(\pi_1^TG^{[m-p]}\right)_j x_{j-t}\\
&=&\sum_{j_0\in\text{Supp}(u^*)}\sum_{j=j_0+t}^{j_0+p}c_jx_{j-t}+\left(\sum_{j_0\in\text{Supp}(u^*)}\sum_{j=j_0+1}^{j_0+t-1}c_j
x_{j-t}+\sum_{j\not\in T\text{
or }j\in\text{Supp}(u^*)}c_j x_{j-t}\right)\\
&\overset{\Delta}{=}& A_t+B_t
\end{eqnarray*}
where the magnitude of $A_t$ can be lower bounded,
\[|A_t|\geq \sum_{j_0\in\text{Supp}(u^*)}(r/2)^{t-1}\beta_{min}(1-2^{-1}-2^{-2}-\cdots-2^{-(p-t)})\geq k\beta_{\min}(r/2)^{p-1}\]
based on property (3) of Lemma~\ref{lem:cond}. And the magnitude of
$B_t$ is upper bounded,
$$|B_t|< \sum_{j_0\in\text{Supp}(u^*)} \frac{\beta_{\max}M\rho^{l-p}}{1-\rho^l}(1+\rho+\rho^2+\cdots)= k\frac{\beta_{\max}M\rho^{l-p}}{(1-\rho^l)(1-\rho)}$$
When $l\geq \left(\log(\frac{2}{1-\rho})+p\log
(\frac{6\beta_{\max}M}{\beta_{\min}})\right)/\log(\rho^{-1})+p$ as
given by the assumption of the theorem, we have $|B_t|<|A_t|$, which
implies that the sign of $(\pi_1^T Y)_{t}$ is determined by the sign
of $A_t$.

Hence $\text{sgn}((\pi_1^T Y)_t)=\text{sgn}(A_t)=1$. This implies
\[\text{sgn}(\pi_1^T Y)=[1,1,\cdots,1]^T\]

In general, for any sign pattern $[s_1,\cdots,s_p]^T
(s_i\in\{-1,1\})$, by choosing $\{c_i\}_{i\in T}$ (compare equation
\ref{eq:pickup_c}) in the following way
\begin{eqnarray*}
  c_i =
  \begin{cases}
  \text{sgn}(u^*_i)& \text{if }i\in
  \text{Supp}(u^*)\\
  s_i\cdot (r/2)^{i-j-1}\text{sgn}(x_j) & \text{if }i=j+1,\cdots,j+p,\forall j\in\text{supp}(u^*)\\
  0 & \text{if }i=j+p+1,\cdots,j+l,\forall j\in\text{supp}(u^*)
 \end{cases}
\end{eqnarray*}
and making similar arguments, we have
\[\text{sgn}(\pi_s^T Y)=[s_1,s_2,\cdots,s_p]^T\]

\section{Proof of Theorem \ref{thm2}}\label{sec:proof2}
To prove Theorem \ref{thm2}, we only need to check that $(\hat
u,\hat a)$ given in the theorem satisfy the KKT conditions. We
denote the function
$f(u,a)=\frac{1}{2}\|y+Ya-u\|_2^2+\lambda\|u\|_1$. Then the gradient
of $f$ with respect to $a$ is
\[\frac{\partial f}{\partial a}=Y^T(y+Ya-u)\]
and the subgradient of $f$ with respect to $u$ is
\[\frac{\partial f}{\partial u}=-(y+Ya-u)+\lambda v\]
where $v$ satisfies $v_i=\text{sgn}(u_i)$ for $i\in I$ and $|v_i|<1$
for $i\in I^c$. Therefore, we only need to check the following set
of (in)equalities
\begin{eqnarray}
Y^T(y+Y\hat a-\hat u)&=&0\label{eq:set1}\\
(y+Y\hat a-\hat u)_i &=& \lambda \text{sgn}(\hat u_i)\label{eq:set2},\quad \hat u_i\not=0\\
|(y+Y\hat a-\hat u)_i| &<& \lambda,\quad\quad\quad\quad  \hat
u_i=0\label{eq:set3}
\end{eqnarray}
We first check Equation \ref{eq:set1}.
\begin{lem}
Equation \ref{eq:set1} is satisfied with $(\hat u, \hat a)$ given in
Theorem \ref{thm2}.
\end{lem}
\begin{proof} Actually,
\begin{eqnarray*}
Y^T(y+Y\hat a-\hat u) &=& Y^T(y-\hat u)+Y^TY\hat a\\
                      &=& Y^T(y-\hat u)-Y^TY(Y^TY)^{-1}Y^T(y-\hat
                      u)=0
\end{eqnarray*}
\end{proof}
Next we check Equation \ref{eq:set2}.
\begin{lem}\label{lemset2}
Equation \ref{eq:set2} is satisfied with $(\hat u, \hat a)$ given in
Theorem \ref{thm2} with probability at least $1-8p/n-(p+1)2^{-n/5}$.
\end{lem}
\begin{proof}
Note that $P$ has the property that $P^2=P$ and $PY=0$. Therefore by
multiplying $P$ to both sides of Equation \ref{eq:noise}, we have
\begin{eqnarray}\label{eq:noise1}
Py=Pu^*+Pe
\end{eqnarray}
Now we can compute
\begin{eqnarray*}
y+Y\hat a-\hat u &=& y-\hat u -Y(Y^TY)^{-1}Y^T(y-\hat u)\\
                 &=& P(y-\hat u)\\
                 &\overset{(a)}{=}& Pu^*+Pe-P\hat u\\
                 &\overset{(b)}{=}& Pe+P_I(u^*_I-\hat u_I) \\
                 &=& Pe-P_I(P^T_IP_I)^{-1}\left(P^T_Ie-\lambda\text{sgn}(u^*_I)\right)
\end{eqnarray*}
where (a) follows from Equation \ref{eq:noise1} and (b) follows from
the fact that $\hat u_{I^c}=u_{I^c}^*=0$.

 There is a small trick
here.  Since $y+Y\hat a-\hat u=P(y-\hat u)$ as we have shown and
$P^2=P$, we must have $P(y+Y\hat a-\hat u)=y+Y\hat a-\hat u$. This
implies Equation \ref{eq:set2} is correct:
\begin{eqnarray*}
(y+Y\hat a-\hat u)_I &=& P_I^T(y+Y\hat a-\hat u)\\
                      &=& P_I^T\left(Pe-P_I(P^T_IP_I)^{-1}\left(P^T_Ie-\lambda\text{sgn}(u^*_I)\right) \right)\\
                      &\overset{(a)}{=}& P_I^TPe-P_I^T e+\lambda\text{sgn}(u^*_I)\\
                      &=&\lambda\text{sgn}(u^*_I)=\lambda\text{sgn}(\hat u_I)
\end{eqnarray*}
where (a) follows from $P_I^T P=P_I^T$ (i.e., $P^2=P$) and the last
equality holds true with probability at least
$1-8p/n-(p+1)2^{-n/5}$. The proof of last equality is similar to the
proof of Lemma \ref{final_lemma} and is omitted here.
\end{proof}

Verifying inequality \ref{eq:set3} requires more effort. We first
simplifies the formula for $(y+Y\hat a-\hat u)_{I^c}$.

\begin{lem}
With $(\hat u,\hat a)$ given in Theorem \ref{thm2}, we have
\begin{eqnarray}
(y+Y\hat a-\hat
u)_{I^c}=-Y_2(Y_2^TY_2)^{-1}Y^{T}_1\lambda\text{sgn}(u^*_I)+(I-Y_2
(Y_2^TY_2)^{-1}Y^T_2)e_{I^c}
\end{eqnarray}
where we denote $Y_1$ as the submatrix comprises of the rows of $Y$
indexed by $I$ and $Y_2$ as the submatrix comprises of the rows of
$Y$ indexed by $I^c$.
\end{lem}
\begin{proof}
Following from the proof of Lemma \ref{lemset2}, we have ,
\begin{eqnarray}\label{eq:nonsupp}
(y+Y\hat a-\hat
u)_{I^c}=P^T_{I^c}(P-P_I(P^T_IP_I)^{-1}P^{T}_I)e+P^T_{I^c}P_I(P^T_IP_I)^{-1}\lambda\text{sgn}(u^*_I)
\end{eqnarray}

To simplify the above equation, we introduce $P_{11}\in
\mathbb{R}^{k\times k}$ as the matrix comprises of the rows of $P$
indexed by $I$ and the columns of $P$ indexed by $I$. Similarly,
$P_{12}\in \mathbb{R}^{(n-k)\times k}$ is the matrix comprises of
the rows of $P$ indexed by $I$ and the columns of $P$ indexed by
$I^c$; $P_{22}\in \mathbb{R}^{(n-k)\times (n-k)}$ is the matrix
comprises of the rows of $P$ indexed by $I^c$ and the columns of $P$
indexed by $I^c$. By this definition, after some column and row
permutations, $P$ can be rewritten as
\begin{eqnarray}
\begin{bmatrix}
P_{11}& P_{12}\\
P_{12}^T & P_{22}
\end{bmatrix}
\end{eqnarray}

It is easy to check that $P_{11}=P_I^TP_I$ and
$P_{12}^T=P^T_{I^c}P_I$ (since $P^2=P$). Furthermore,
\begin{eqnarray*}
P_{I^c}^T-P^T_{12}P^{-1}_{11}P^{T}_I&=&\begin{bmatrix}P_{12}^T &
P_{22}\end{bmatrix}-P^T_{12}P^{-1}_{11}\begin{bmatrix}P_{11}&
P_{12}\end{bmatrix}\\
&=&\begin{bmatrix}0& P_{22}-P_{12}^T P_{11}^{-1}P_{12}\end{bmatrix}
\end{eqnarray*}

Hence, Equation \ref{eq:nonsupp} can be simplified to
\[(y+Y\hat a-\hat
u)_{I^c}=(P_{22}-P_{12}^T P_{11}^{-1}P_{12})e_{I^c}+\lambda
P^T_{12}P^{-1}_{11}\text{sgn}(u^*_I)\]

We note that $P_{11},P_{12},P_{22}$ can be expressed in terms of $Y,
Y_1$ and $Y_2$.
\begin{eqnarray*}
P_{11} &=& I-Y_1 (Y^TY)^{-1}Y^T_1 \\
P_{12} &=& -Y_1 (Y^TY)^{-1}Y^T_2\\
P_{22} &=& I-Y_2 (Y^TY)^{-1}Y^T_2
\end{eqnarray*}
Moreover $P_{11}^{-1}$ can be derived via matrix inversion lemma:
\begin{eqnarray*}
P_{11}^{-1} &=& (I-Y_1 (Y^TY)^{-1}Y^T_1)^{-1} \\
&=& I+Y_1(Y^TY-Y_1^TY_1)^{-1}Y^{T}_1=I+Y_1(Y_2^TY_2)^{-1}Y^{T}_1
\end{eqnarray*}
Finally, we get
\begin{eqnarray*}
\lambda P^T_{12}P^{-1}_{11}\text{sgn}(u^*_I) &=& -
Y_2(Y^TY)^{-1}Y^T_1(I+Y_1(Y_2^TY_2)^{-1}Y^{T}_1)\lambda\text{sgn}(u^*_I)\\
  &=& -Y_2
  [(Y^TY)^{-1}+(Y^TY)^{-1}Y_1^TY_1(Y_2^TY_2)^{-1}]Y^{T}_1\lambda\text{sgn}(u^*_I)\\
  &=&-Y_2
  (Y^TY)^{-1}[Y_2^TY_2+Y_1^TY_1](Y_2^TY_2)^{-1}Y^{T}_1\lambda\text{sgn}(u^*_I)\\
  &\overset{(a)}{=}&-Y_2(Y_2^TY_2)^{-1}Y^{T}_1\lambda\text{sgn}(u^*_I)
\end{eqnarray*}
where (a) follows from the fact that $Y^TY = Y_2^TY_2+Y_1^TY_1$. And
similarly by repeatedly using this fact we can find the following
simplification
\begin{eqnarray*}
(P_{22}-P_{12}^T P_{11}^{-1}P_{12})e_{I^c} &=& (I-Y_2
(Y^TY)^{-1}Y^T_2)e_{I^c}\\
& &-(Y_2(Y^TY)^{-1}Y^T_1(I+Y_1(Y_2^TY_2)^{-1}Y^{T}_1)Y_1
(Y^TY)^{-1}Y^T_2)e_{I^c}\\
&=&(I-Y_2 (Y_2^TY_2)^{-1}Y^T_2)e_{I^c}
\end{eqnarray*}
\end{proof}
In order to justify the condition \ref{eq:set3}, we also need the
following lemma.
\begin{lem}\label{claim3}
The following three claims hold true:
\begin{description}
  \item[(i)] w.p. at least $1-p\cdot(4/n+2^{-n/5})$, $\|Y_2^T e_{I^c}\|_\infty\leq 2\sqrt{n\log n}\sigma a_{\max}x_{\max}\sqrt{2p}$.
  \item[(ii)] w.p. at least $1-4p/n$, $\|Y_1^T \lambda \text{sgn}(z_I^*)\|_\infty\leq 2\lambda\|x\|_2\sqrt{\log n}$
  \item[(iii)] w.p. at least $1-2^{-n/5}$, $\lambda_{\max}\left((Y_2^TY_2)^{-1}\right)\leq
  2\lambda_{\max}\left((X_2^TX_2)^{-1}\right)\leq
  \frac{2c}{\|x\|_2^2}$
\end{description}
\end{lem}
\begin{proof}
To prove (i), we try to bound the first component
 $\left(Y_2^T e_{I^c}\right)_1$. By definition, the first column of
 $Y$ equals $[0,\, y_0, \cdots,y_{n-2}]^T=[0,\, x_0, \cdots,x_{n-2}]^T+[0,\, w_0,
 \cdots,w_{n-2}]^T$. We also remember $e_i =
 w_i+\sum_{j=1}^pa_jw_{i-j}$ where $w_i$ are i.i.d. Gaussian
 $\mathcal{N}(0,\sigma^2)$. Hence, we have
 \begin{eqnarray*}
\left(Y_2^T e_{I^c}\right)_1 = \sum_{i\in
I^c}x_{i-1}(w_i+\sum_{j=1}^pa_jw_{i-j})+\sum_{i\in
I^c}w_{i-1}(w_i+\sum_{j=1}^pa_jw_{i-j})
 \end{eqnarray*}

 It is easy to check that the first term of RHS is zero-mean Gaussian
random variable with variance $\leq pa_{\max}^2x_{\max}^2n\sigma^2$.
It is well known that for standard Gaussian random variable $t$,
$\Pr(|t|\geq a)\leq 2e^{-a^2/2}$. So we conclude that with
probability $\geq 1-2/n$
$$\left|\sum_{i\in
I^c}x_{i-1}(w_i+\sum_{j=1}^pa_jw_{i-j})\right|\leq \sigma
a_{\max}x_{\max}\sqrt{2pn\log n}$$ It also can be proved that with
probability $\geq 1-2/n-2^{-n/5}$
$$\left|\sum_{i\in
I^c}w_{i-1}(w_i+\sum_{j=1}^pa_jw_{i-j})\right|\leq 2pa_{\max}
\sigma^2\sqrt{n\log n}$$ We notice that $\sigma
a_{\max}x_{\max}\sqrt{2pn\log n}\geq 2pa_{\max} \sigma^2\sqrt{n\log
n}$ and hence claim (i) follows.

Next, we prove claim (ii). Again, $\left(Y_1^T \lambda
\text{sgn}(z_I^*)\right)_1$ can be decomposed into two terms;
\begin{eqnarray*}
\left(Y_1^T \lambda \text{sgn}(z_I^*)\right)_1 = \left(X_1^T \lambda
\text{sgn}(z_I^*)\right)_1+\sum_{i\in I^c}w_{i-1}\lambda
\text{sgn}(z_I^*)
 \end{eqnarray*}
 The first term is bounded from the assumption and the second term is
 Gaussian which is bounded by $\lambda\sigma\sqrt{2n\log n}\leq \lambda\|x\|_2\sqrt{\log n}$ (assumption (3) in Subsection \ref{sec:blind}) w.p. $\geq
 1-2/n$.

For (iii), we only need to show that with high probability
$\lambda_{\min}\left(Y_2^TY_2\right)\geq
  \frac{1}{2}\lambda_{\min}\left(X_2^TX_2\right)$, or $\sigma_{\min}\left(Y_2\right)\geq
  \frac{1}{\sqrt{2}}\sigma_{\min}\left(X_2\right)$ where $
  \sigma_{
  \min}(A)$ denotes the smallest singular value of $A$.

We denote the Gaussian noise matrix
 $$W = \begin{bmatrix}
 0 & \cdots & 0\\
 w_0 & \cdots &0\\
  \vdots & \ddots & \vdots\\
w_{p-1} &\cdots & w_0\\
 \vdots & \vdots & \vdots\\
 w _{n-2} & \cdots & w_{n-p}
\end{bmatrix}$$
and call $W_2$ as the submatrix that comprises of the rows of $W$
indexed by $I^c$.  Then, we have
\begin{eqnarray*}
\sigma_{\min}\left(Y_2\right) &=& \min_{\|t\|_2=1}\|Y_2 t\|_2 =
\min_{\|t\|_2=1}\|X_2t+W_2 t\|_2\\
&\geq& \min_{\|t\|_2=1}\|X_2t\|_2-\max_{\|t\|_2=1}\|W_2 t\|_2 =
\sigma_{\min}(X_2)- \sigma_{\max}(W_2)
\end{eqnarray*}

So the remaining work is to upper bound $\sigma_{\max}(W_2)$. A
tight bound in this case is very difficult. However, the following
bound is good enough for our proof. By denoting $W_{2,i}$ as the
$i$-th column of $W_2$, we have
\begin{eqnarray*}
\sigma_{\max}(W_2) &=& \max_{\|t\|_2=1}\|W_2 t\|_2\\
                   &=& \max_{\|t\|_2=1}\sqrt{\sum_{i}\langle W_{2,i},
                   t\rangle^2}\\
                   &\leq & \sqrt{\sum_{i}\|W_{2,i}\|_2^2}\leq
                   \sqrt{p\|w\|_2^2}
\end{eqnarray*}
where the second last inequality follows from Cauchy-Schwartz
inequality. Then by the tail probability of $\chi^2$ distribution,
we have with probability $1-2^{-n/5}$,
\begin{eqnarray*}
\sigma_{\max}(W_2) \leq
                   \sqrt{p\|w\|_2^2}\leq \sqrt{2np\sigma^2}
\end{eqnarray*}
Then by applying assumption (1) in Subsection \ref{sec:blind} we
have proved the claim (iii).
\end{proof}

Finally, we can show that $(\hat u,\hat a)$ satisfies the condition
\ref{eq:set3}.
\begin{lem}\label{final_lemma}
Equation \ref{eq:set1} is satisfied with $(\hat u, \hat a)$ given in
Theorem \ref{thm2} with probability at least $1-8p/n-(p+1)2^{-n/5}$.
\end{lem}
\begin{proof} From the tail probability of standard Gaussian $\Pr(|t|\geq a)\leq 2e^{-a^2/2}$, we know that
with probability at least $1-2/n$, $\max_{i}|w_i|\leq
2\sigma\sqrt{\log n}$. Therefore the $\ell_2$ norm of all the rows
of $Y_2$ is upper bounded $\sqrt{p}(x_{\max}+2\sigma\sqrt{\log n})$
with probability at least $1-2/n$.  Combined with claim (iii) in
Lemma \ref{claim3}, we know that the $\ell_2$ norm of all the rows
of $Y_2(Y_2^T Y_2)^{-1}$ is upper bounded
$\frac{2c\sqrt{p}}{\|x\|_2^2}(x_{\max}+2\sigma\sqrt{\log n})\leq
\frac{4c\sqrt{p}\,x_{\max}}{\|x\|_2^2}$ with probability at least
$1-2/n-2^{-n/5}$.

Now we can verify that both
$-Y_2(Y_2^TY_2)^{-1}Y^{T}_1\lambda\text{sgn}(u^*_I)$ and $(I-Y_2
(Y_2^TY_2)^{-1}Y^T_2)e_{I^c}$ are small.

First, based on claim (ii) in Lemma \ref{claim3}, with probability
at least $1-\frac{2+4p}{n}-2^{-n/5}$
\begin{eqnarray*}
\|-Y_2(Y_2^TY_2)^{-1}Y^{T}_1\lambda\text{sgn}(u^*_I)\|_\infty\leq
\frac{4c\sqrt{p}\,x_{\max}}{\|x\|_2^2}\cdot2\lambda\|x\|_2\sqrt{\log
n}\cdot\sqrt{p}<\lambda/3
\end{eqnarray*}
where the last inequality follows from condition (3) in Subsection
\ref{sec:blind}.

 Next, it
is easy to bound $\|e_{I^c}\|_\infty\leq 2\sigma pa_{\max}\sqrt{\log
n}\leq \lambda/3$ with probability at least $1-2/n$. Also, we have
with probability at least $1-\frac{4p+2}{n}-(p+1)2^{-n/5}$
\begin{eqnarray*}
\|Y_2 (Y_2^TY_2)^{-1}Y^T_2 e_{I^c}\|_\infty\leq
\frac{4c\sqrt{p}\,x_{\max}}{\|x\|_2^2}\cdot2\sqrt{n\log n}\sigma
a_{\max}x_{\max}\sqrt{2p}\cdot\sqrt{p}<\lambda/3
\end{eqnarray*}
where the last inequality follows from claim (i) of Lemma
\ref{claim3}, condition (3) in Subsection \ref{sec:blind} and the
assumption $\lambda\geq 6\sigma p a_{\max}\sqrt{\log n}$.
\end{proof}

\section{Numerical Experiments}\label{sec:simulation}
We present simulations for some interesting cases. Theorem \ref{thm1} asserts that as
long as RIP is satisfied, stability assumptions on $H$ hold, and the spikes are well separated, our $\ell_1$-minimization
algorithm reconstructs the AR process correctly. For general
IID Gaussian or Bernoulli matrix ensemble (not Toeplitz), it is well known
that \cite{Candes3} $m\geq O(S\log (n/S))$ ensures good RIP
property. However, for our specific Toeplitz structured sensing
matrix (Equation \ref{eq:sense}), this question (when RIP is
satisfied) has not been fully answered.
%We
%illustrate the performance of this Toeplitz-structured sensing
%matrix through simulations.

We nevertheless experiment with Toeplitz constructions. First we simulate our algorithm for a third order process. The results are depicted in Figure~\ref{fig:recons}. We see that the reconstruction reproduces both the spike train as well as the filtered process accurately. For the purpose of depiction we added a small amount of noise.
\begin{figure}[t]
\begin{centering}
\begin{minipage}[t]{.48\linewidth}
\includegraphics[width = 4.0cm]{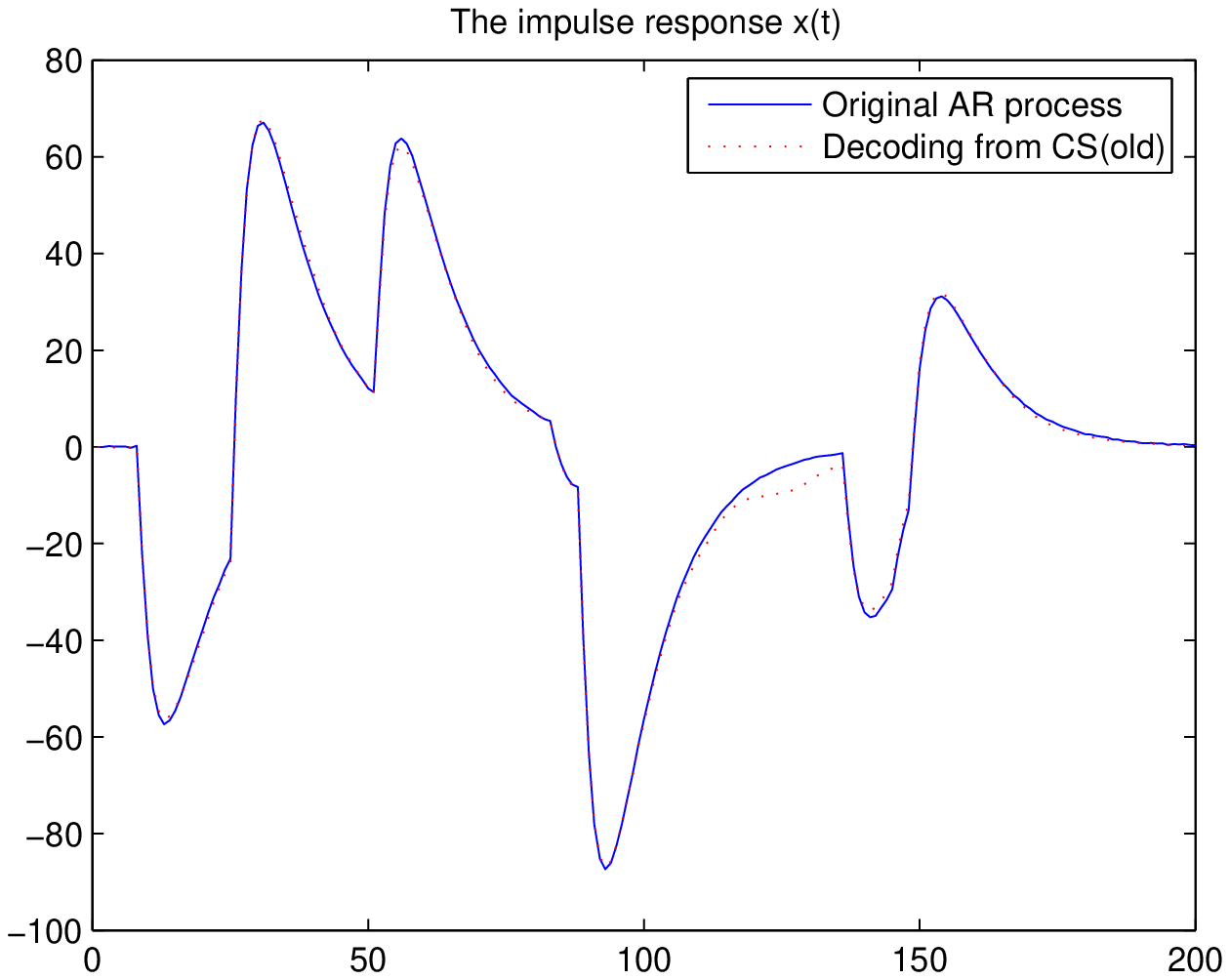}\\
%\makebox[7 cm]{(a)}
%\centering
%  \centerline{\epsfig{figure=../ARprocess/tap2_50.eps,width=4.0cm}}
\end{minipage}
\begin{minipage}[t]{.48\linewidth}
\includegraphics[width = 4.0cm]{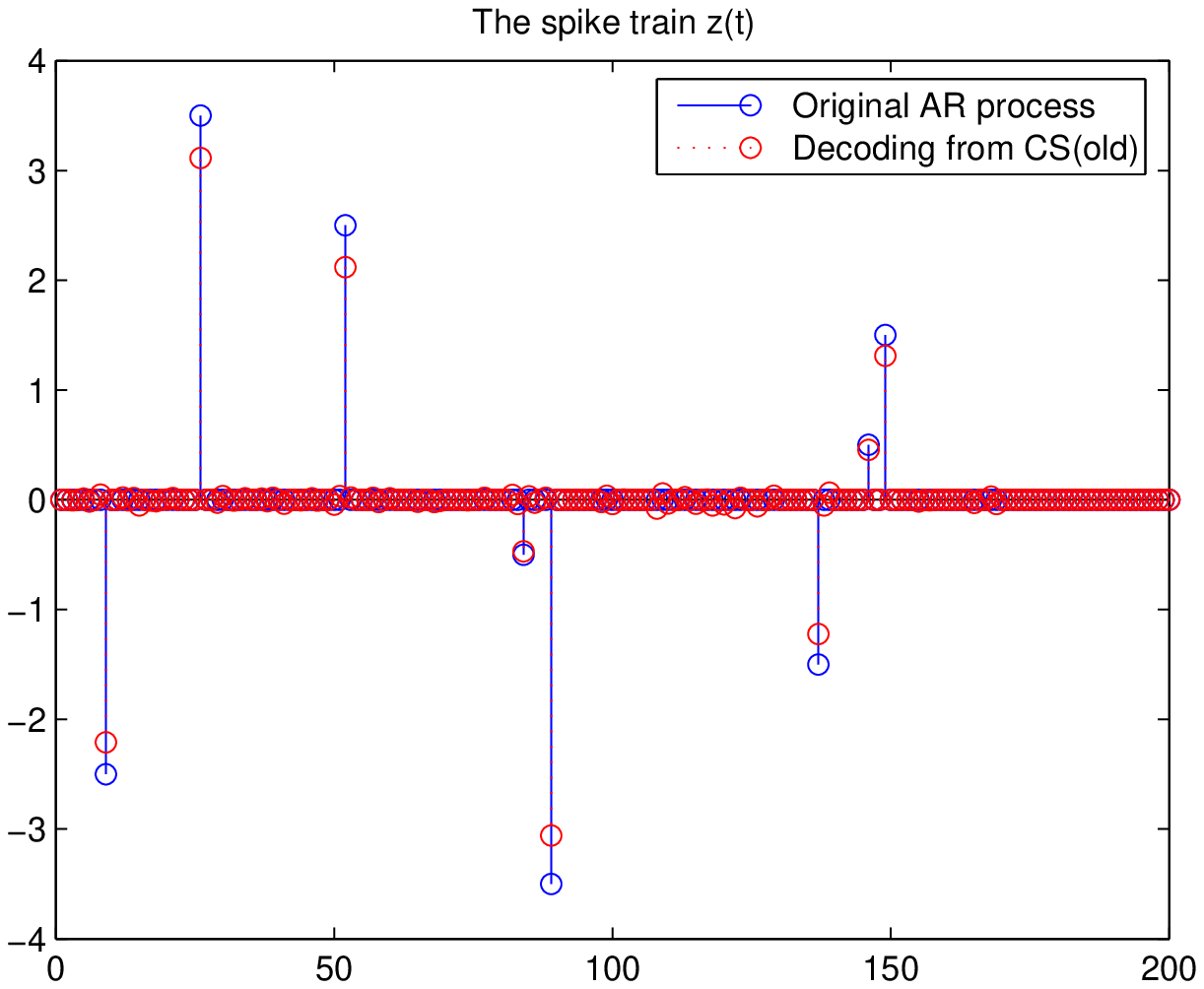}\\
%\makebox[7 cm]{(b)}
%\centering
%  \centerline{\epsfig{figure=../ARprocess/tap2_50_ber.eps,width=4.0cm}}
\end{minipage}
\caption{\small $\ell_1$-minimization algorithm on the model $y=Gx + w$
with $G$ an $80\times 200$ Toeplitz Gaussian matrix ensemble. The filtered process $x(n)$ is obtained by filtering a $8$ sparse spike train through a  third-order AR process with poles $\alpha_1=0.9$, $\alpha_2=0.5$ and $\alpha_3=0.2$. The measurements were contaminated with zero mean Gaussian noise with variance $0.1$.} \label{fig:recons}
\end{centering}
\end{figure}

\begin{figure}[t]
\begin{centering}
\begin{minipage}[t]{.48\textwidth}
\includegraphics[width = 1\textwidth]{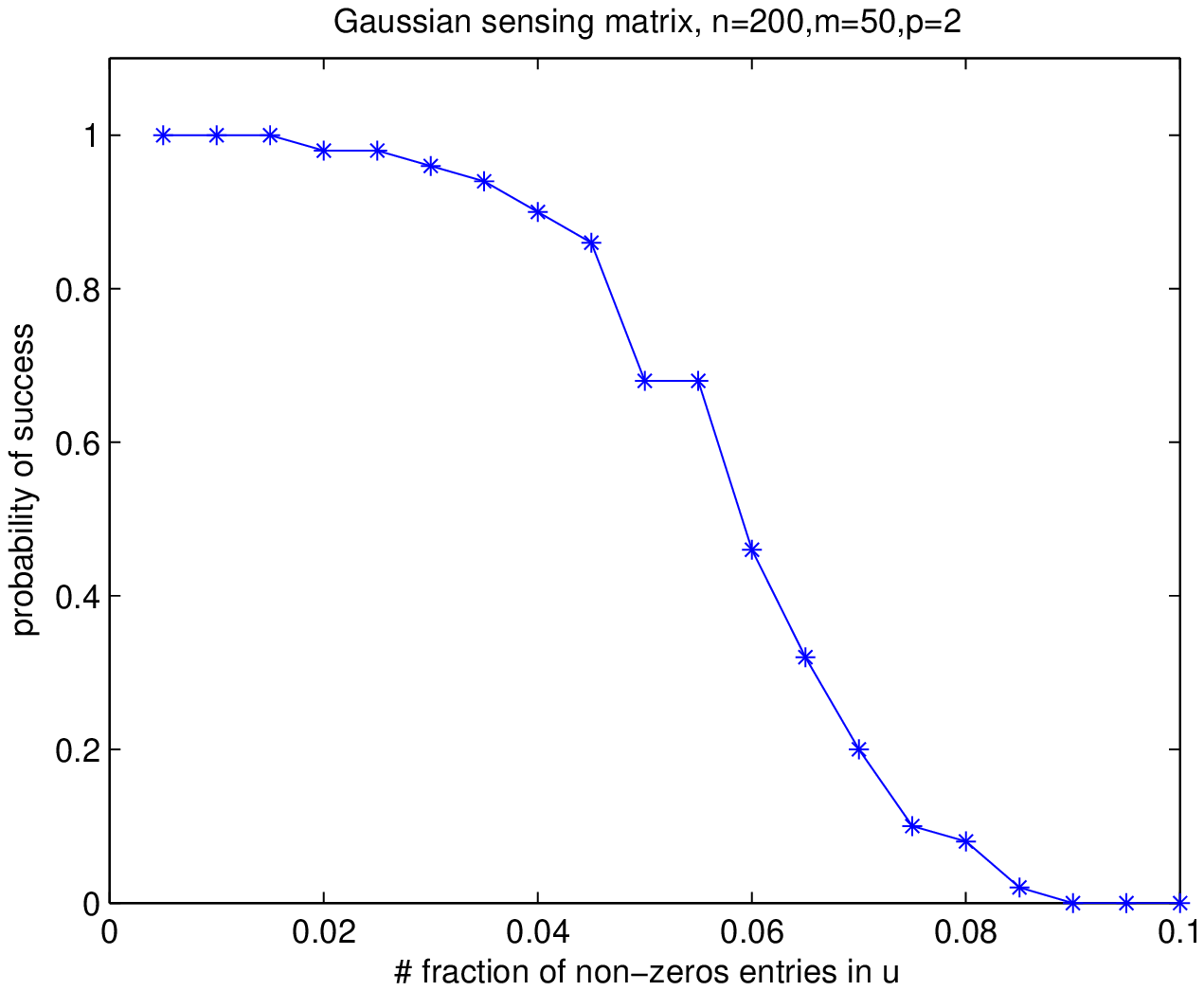}\\
\makebox[7 cm]{(a)}
\end{minipage}
\begin{minipage}[t]{.48\textwidth}
\includegraphics[width = 1\textwidth]{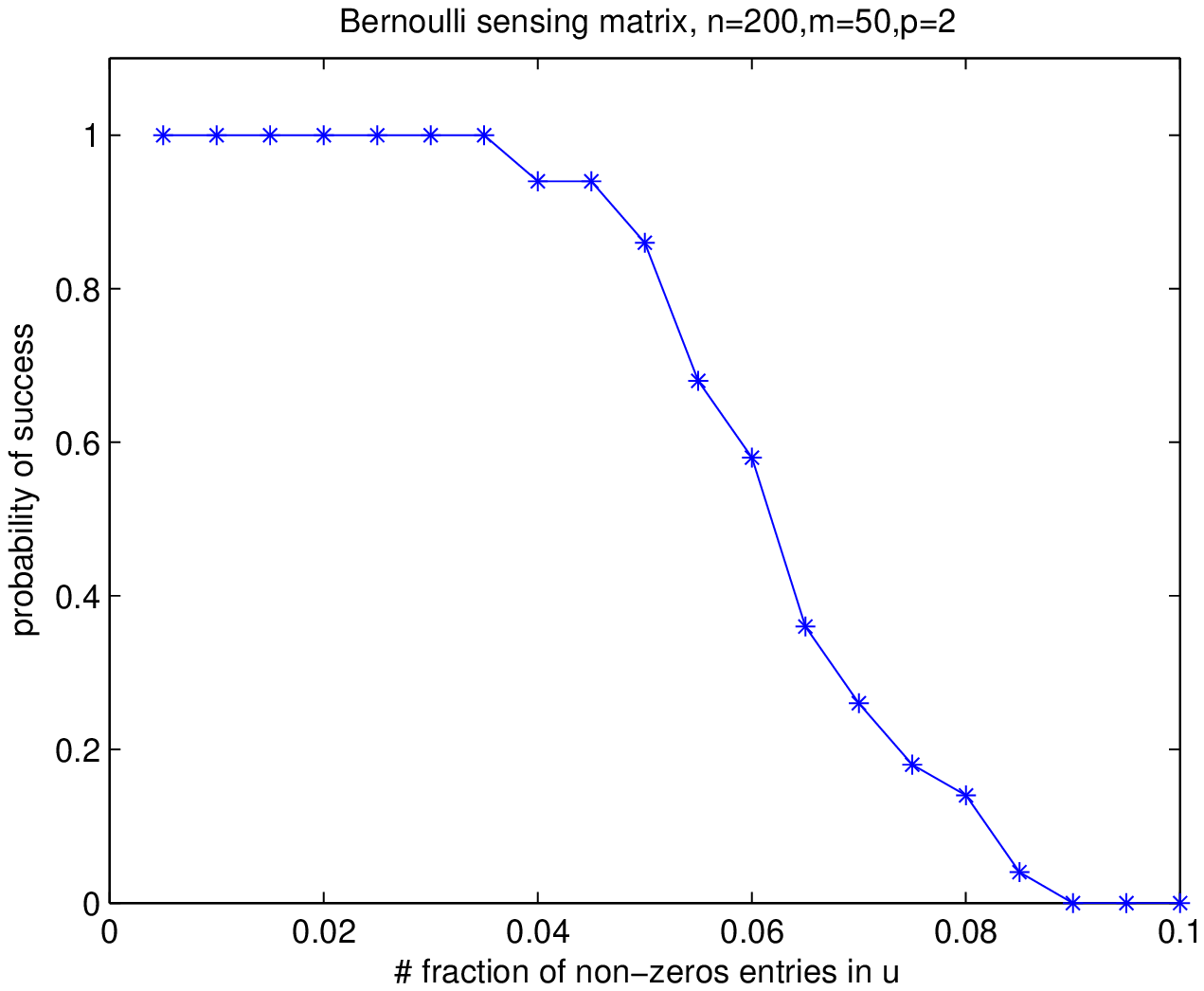}\\
\makebox[7 cm]{(b)}
\end{minipage}
\caption{\small $\ell_1$-minimization algorithm on the model $y=Gx$
with $G$ an $50\times 200$ Toeplitz matrix with independent Gaussian
or Bernoulli entries. In this experiment $x(n)$ is a second-order AR
process with poles $\alpha_1=0.9$ and $\alpha_2=0.5$; (a) success
rate when $G$ is Gaussian $\mathcal{N}(0,1)$; (b) success rate when
$G$ is Bernoulli $\pm 1$.} \label{fig:cutoff}
\end{centering}
\end{figure}

First, we fix the size of sensing matrix ($m=50, n=200$) and choose
the entries of sensing matrix $G$ to be Gaussian. We also fix the
order of the AR model ($p=2$) and let the sparsity $k$ vary from $1$
to $20$. For each fixed $k$, we run our $\ell_1$-minimization
algorithm $50$ times to obtain the average performance. The result is shown in Figure \ref{fig:cutoff}(a).
Similarly, we can choose the sensing matrix $G$ to be Bernoulli $\pm
1$ and do the same experiment again. The result is shown in Figure
\ref{fig:cutoff}(b). We can see that in this example Toeplitz
Bernoulli
 matrix is more preferable than Toeplitz Gaussian matrix.

Next, we run our algorithm on a case that does not satisfy our assumptions on stability. Specifically we consider the situation when the true process  is governed by the equation $x(n)-x(n-1)=u(n)$. This type of model is closely associated with problems that arise when one is interested in minimizing total variations. Note that in this model $\alpha=1$ and it
does not satisfy the assumptions of Theorem \ref{thm1} where we
assume $\alpha_{\max}<1$. We adopt the same sensing matrix
as the last experiment (Gaussian or Bernoulli) and the empirical
success rate of this experiment is shown in Figure
\ref{fig:cutoff2}.
\begin{figure}[t]
\begin{centering}
\begin{minipage}[t]{.48\textwidth}
\includegraphics[width = 1\textwidth]{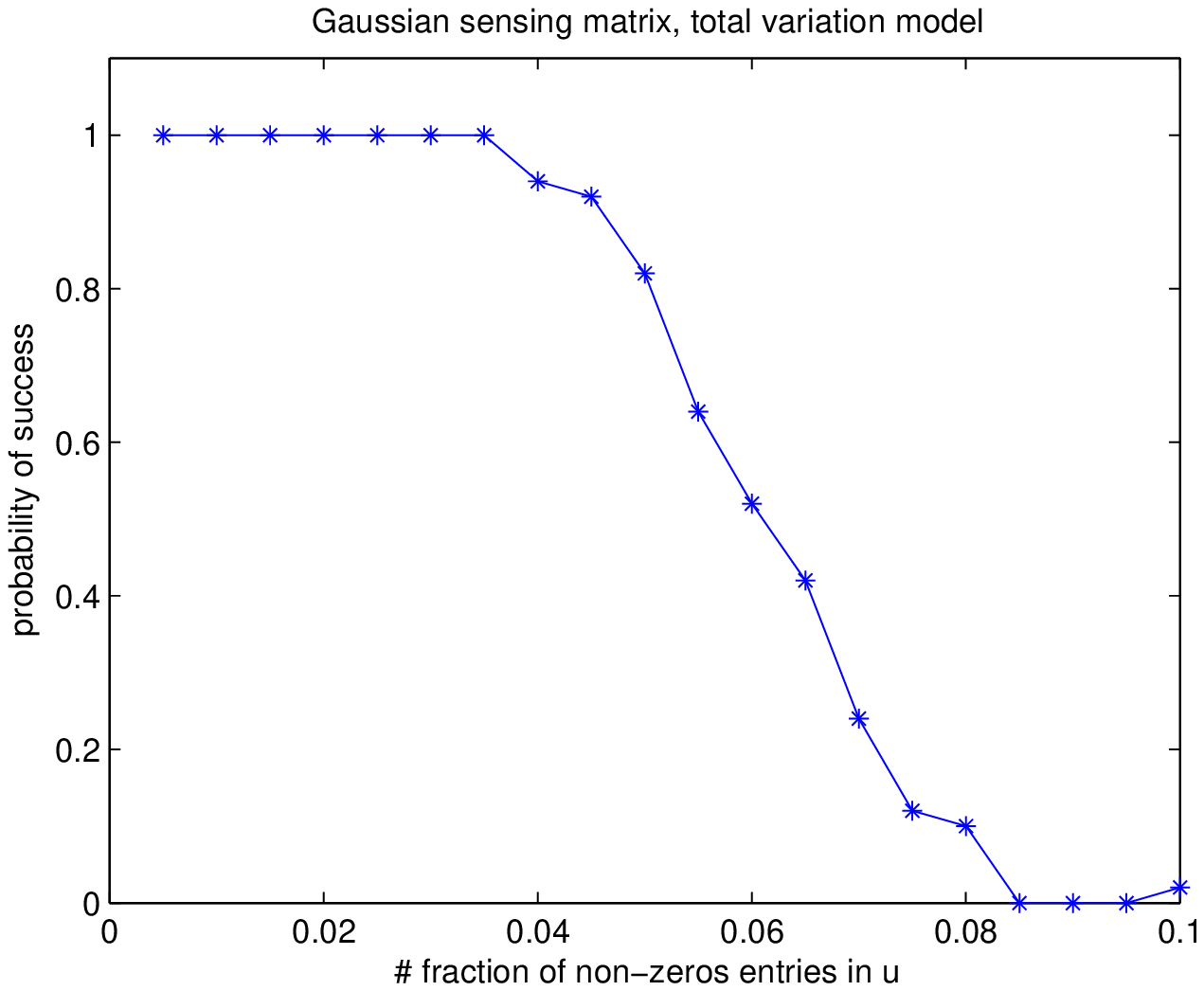}\\
\makebox[7 cm]{(a)}
\end{minipage}
\begin{minipage}[t]{.48\textwidth}
\includegraphics[width = 1\textwidth]{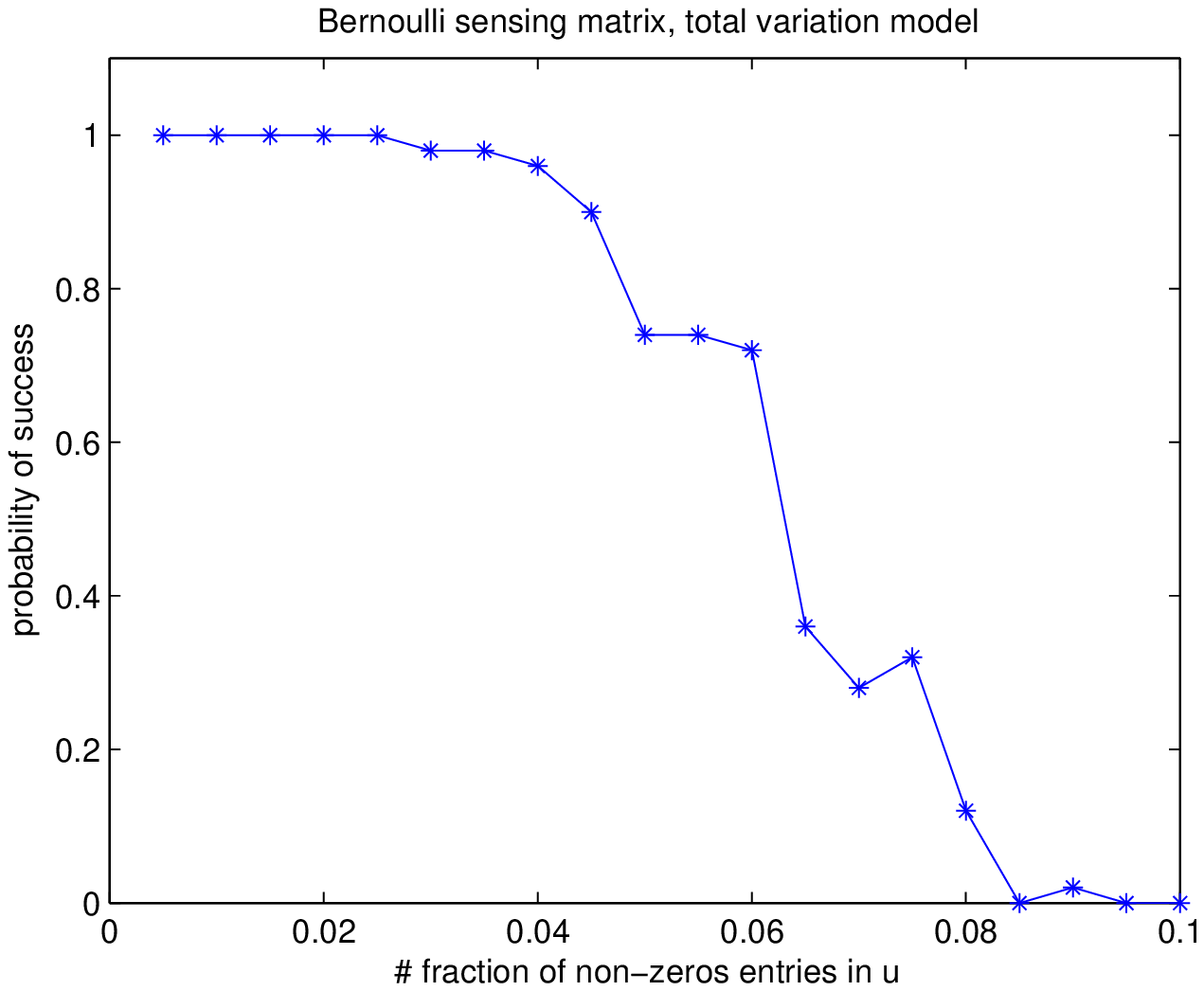}\\
\makebox[7 cm]{(b)}
\end{minipage}
\caption{\small $\ell_1$-minimization algorithm on the model $y=Gx$
with $G$ an $50\times 200$ Toeplitz matrix with independent Gaussian
or Bernoulli entries. In this experiment $x(n)$ is total variation
process $x(n)-x(n-1)=u(n)$; (a) success rate when $G$ is Gaussian
$\mathcal{N}(0,1)$; (b) success rate when $G$ is Bernoulli $\pm 1$.}
\label{fig:cutoff2}
\end{centering}
\end{figure}

\begin{figure}[t]
\begin{centering}
\begin{minipage}[t]{.48\textwidth}
\includegraphics[width = 1\textwidth]{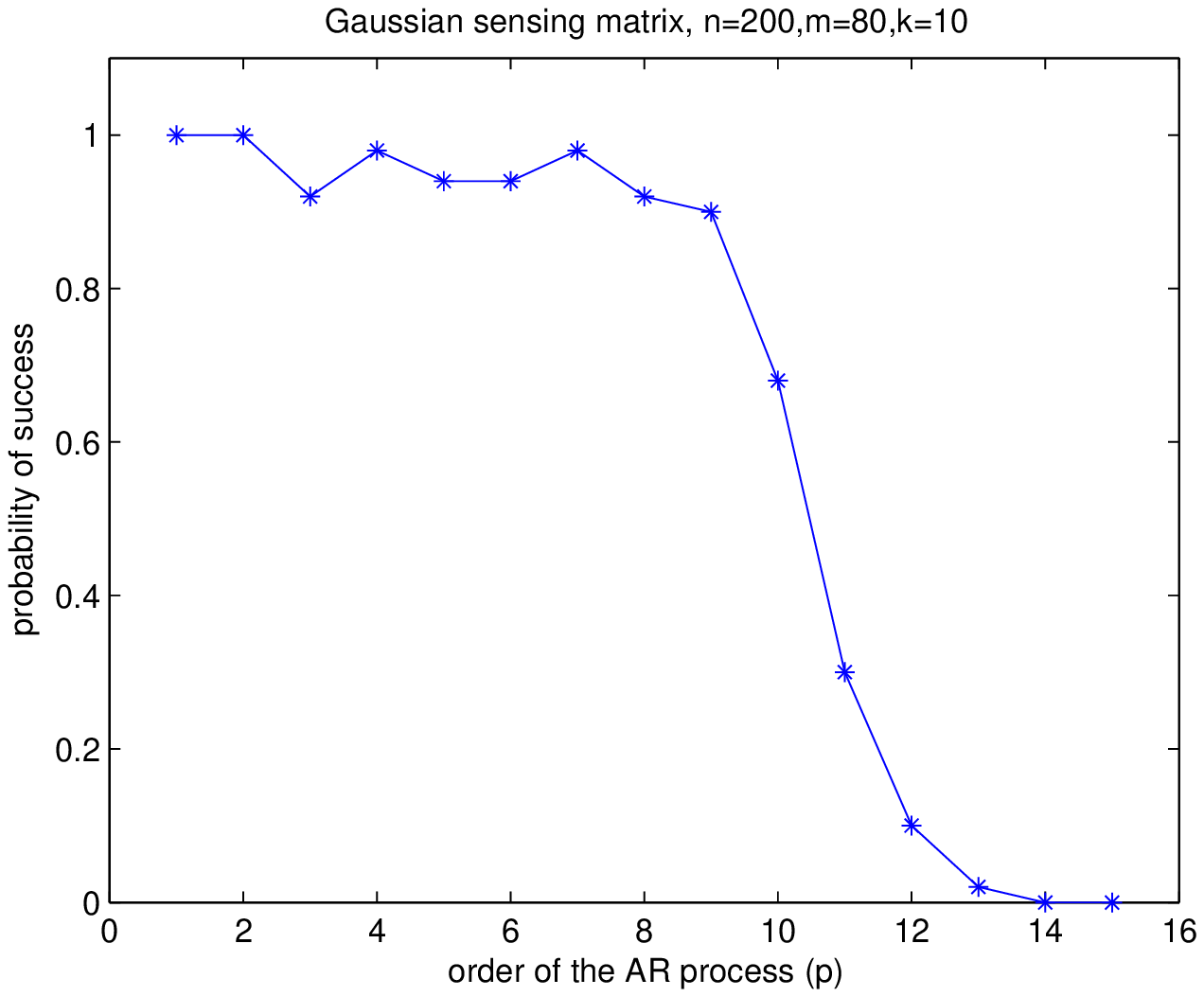}\\
\makebox[7 cm]{(a)}
\end{minipage}
\begin{minipage}[t]{.48\textwidth}
\includegraphics[width = 1\textwidth]{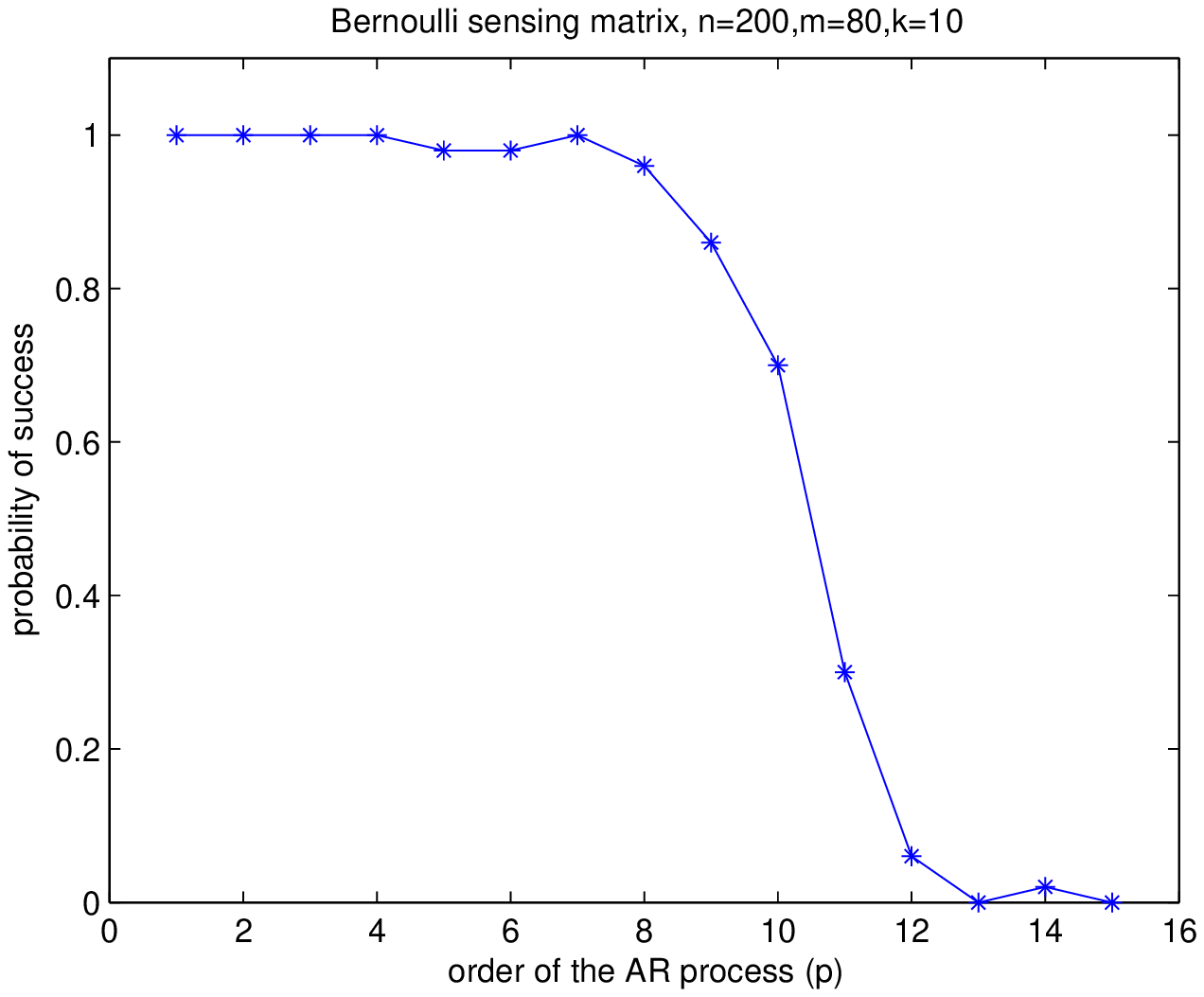}\\
\makebox[7 cm]{(b)}
\end{minipage}
\caption{\small $\ell_1$-minimization algorithm on the model $y=Gx$
with $G$ an $80\times 200$ Toeplitz matrix with independent Gaussian
or Bernoulli entries. In this experiment the order of the AR process
$x(n)$ is a variable, i.e., $p\in [1,15]$; (a) success rate when $G$
is Gaussian $\mathcal{N}(0,1)$; (b) success rate when $G$ is
Bernoulli $\pm 1$.} \label{fig:cutoff3}
\end{centering}
\end{figure}

Finally we test how the order of the AR process influences
the performance of the algorithm. In this experiment, we fix the
size of the sensing matrix as $80\times 200$ and also fix the
sparsity $k=10$ (i.e., the $\#$ fraction of nonzero components in
$z$ is $5\%$). We let $p$ (order of the AR process) vary from $1$ to
$15$. Figure \ref{fig:cutoff3}(a) shows that empirical success rate
for the Gaussian sensing matrix and Figure \ref{fig:cutoff3}(b)
shows that success rate for the Bernoulli sensing matrix. We can see
that again Bernoulli Toeplitz matrix outperforms  the Gaussian
Toeplitz matrix.

\bibliographystyle{IEEEtran}
\bibliography{ARCS}

\end{document}